\documentclass[a4paper,10pt,twoside]{article}
\usepackage[top=3cm, bottom=3cm, inner=3cm, outer=3cm, includehead]{geometry}
\usepackage{fancyhdr}
\usepackage{xurl}
\usepackage{graphicx}
\usepackage{alltt}
\usepackage{amsmath}
\usepackage[hidelinks, pdftex]{hyperref}
\usepackage[T1]{fontenc}
\usepackage[utf8]{inputenc}
\usepackage{lmodern}
\usepackage{csquotes}
\usepackage[style=numeric, sorting=none]{biblatex}
\usepackage{lipsum}

\newcommand\blfootnote[1]{%
  \begingroup
  \renewcommand\thefootnote{}\footnote{#1}%
  \addtocounter{footnote}{-1}%
  \endgroup
}

\bibliography{neurips_2025}

\usepackage[english]{babel}

\usepackage[utf8]{inputenc} 
\usepackage[T1]{fontenc}    
\usepackage{hyperref}       
\usepackage{url}            
\usepackage{booktabs}       
\usepackage{amsfonts}       
\usepackage{nicefrac}       
\usepackage{microtype}      
\usepackage{xcolor}         
\usepackage{comment}
\usepackage{amsmath,amsthm,amssymb}
\usepackage{graphicx}
\usepackage{subcaption}
\newtheorem{lemma}{Lemma}

\begin{document}
\begin{center}
\LARGE
\textbf{On the Dynamics of Acceleration in First order Gradient Methods}\\[12pt]
\normalsize
\vspace{0.4cm}
\textbf {M Parimi,\footnote{Research Scholar, $E-MC^2$ Lab,Veermata Jijabai Technological Institute (VJTI), Mumbai, India} Rachit Mehra,\footnote{Project Lead, TenneT Offshore GmBH \{rachit.mehra\}@tennet.eu}, S.R. Wagh\footnote{Faculty- EED \& P.I. ($E-MC^2$), VJTI}, Amol Yerudkar, \footnote{Associate Professor, School of Computer Science and Technology, Zhejiang Normal University} and Navdeep Singh \footnote{IGI Research Chair Professor, $E-MC^2$ Lab, VJTI} \blfootnote{$E-MC^2$ Lab, VJTI, acknowledges International Gemological Institute (IGI) for financially supporting this research and Savex Technologies for establishing the lab and providing research facilities.}}\\[4pt]
\end{center}
\vspace{0.4cm}
\begin{abstract}
\normalsize
Ever since the original algorithm by Nesterov (1983), the true nature of the acceleration phenomenon has remained elusive, with various interpretations of why the method is actually faster. The diagnosis of the algorithm through the lens of Ordinary Differential Equations (ODEs) and the corresponding dynamical system formulation to explain the underlying dynamics has a long and rich history. In the literature, the ODE’s that explain algorithms are typically derived by considered the limiting case of the algorithm maps themselves, that is, an ODE formulation follows the development of an algorithm. This obfuscates the underlying higher order principles behind the generation of the ODE and thus provides little or no evidence of the working of the algorithm. Such has been the case with Nesterov’s algorithm (NA) and the various analogies used to describe the acceleration phenomena, viz, momentum associated with the rolling of a Heavy-Ball down a slope, damping / Hessian damping etc. This is why the working of the algorithm continues to be shrouded in mystery, even after over 40 years since it was first proposed.\\ 
\indent The main focus of our work is to ideate the genesis of the Nesterov algorithm from the viewpoint of  dynamical systems leading to demystifying the mathematical rigour behind the algorithm. Instead of reverse - engineering ODE’s from discrete algorithms (as is standard), this work explores tools from the recently developed control paradigm titled "The Passivity and Immersion (P\&I) approach" and the Geometric Singular Perturbation theory which are applied to arrive at the formulation of a dynamical system that explains and models the acceleration phenomena. This perspective facilitates to gain insights into the various terms present and the sequence of steps used in Nesterov’s accelerated algorithm for the smooth strongly convex and the convex case. The framework can also be extended to derive the acceleration achieved using the triple momentum method and  further provides justifications for the non-convergence to the optimal solution in the Heavy-Ball method.\vskip 2mm
\textbf{Keywords:} Nesterov's accelerated algorithm, Time scale mixing, Geometric singular perturbation theory, Fenichel's theorem, Invariant manifolds, Dynamical systems, Slow-fast systems, Normally attractive invariant manifold, Optimization, Passivity and Immersion approach
\end{abstract}

\section{Introduction and Motivation}
The challenges in obtaining global solutions to unconstrained optimization of large scale systems (machine learning, engineering optimization etc) involve high dimensionality, nonlinear objective functions, sparsity in data etc leading to ill-conditioned problems. These factors lead to slow convergence in gradient based solvers, prompting the use of adaptive optimizers, preconditioning and second order methods. Though second order methods (such as Newton's method) incorporate the Hessian information to address ill-conditioning, they suffer from high computational costs and memory requirements. An approach involving faster convergence using gradient descent with the addition of the so called "look ahead term" was proposed by Nesterov in his seminal work \cite{nesterov2013introductory}.\\
Nesterov’s algorithm (NA) attains accelerated convergence comparable to second-order methods, despite relying solely on first-order information (also called the Oracle). Efforts to unravel this phenomenon have been abundant in the literature, drawing on concepts such as Hessian damping, the Heavy-ball method, and other momentum-based techniques. Below are highlighted some of the reasons why NA is considered mysterious:
\begin{enumerate}
\item 
    Lack of a clear geometric intuition: Unlike standard gradient descent, which can be easily visualized as a ball rolling down a hill, Nesterov's method doesn't have a simple, direct physical analogy. The "Heavy-Ball" method, a related but distinct algorithm, has a more intuitive mechanical interpretation as a mass on a frictionless surface. NA, with its "look-ahead" step, is more difficult to visualize. The look-ahead step is one of the key differences from the heavy-ball method, as it calculates the gradient at a projected future position rather than the current one. This can feel unnatural and obscure.
    \item It's not a descent method: A key feature of standard gradient descent is that the objective function value decreases at every step. Nesterov's method, however, is not a monotonic descent method. The objective function value can oscillate and even increase for some iterations, before ultimately converging at a faster rate. This is contrary to the usual intuition that a successful optimization algorithm should always be making progress toward a lower function value.
    \item The proof is not enlightening: The original proof provides a solid algebraic guarantee of convergence but doesn't offer a deep understanding of the mechanism of acceleration. It's a "certificate of correctness" rather than a pedagogical explanation. 
    This has motivated decades of research in pursuit of a clearer and more intuitive explanation, which are reviewed in the next section.
\end{enumerate}
\subsection {Survey of Related Work}
Since Nesterov’s 1983 paper, researchers have sought to understand why acceleration is possible, aiming to move beyond the elegant yet so called ``tricky" algebraic manipulations of the original proof. In recent years, this quest for a deeper explanation has gained momentum, particularly driven by the machine learning community’s interest in generalizing Nesterov’s acceleration to broader application domains. However, such generalizations remain challenging, especially given that the fundamental mechanism behind the acceleration is still not fully understood. Recent works in the “why” direction are reinterpreting the NA from different point of view, and can broadly be categorized as follows:
\begin{enumerate}
\item Continuous-time perspective: The work of \cite{su2016differential} highlights the continuous-time ODE representation of the NA, with a second order ODE consisting of a momentum and a time-dependent friction, indicating a large damping ratio for a small $t$. With increase in $ t$, the decrease in the damping ratio causes the system from being overdamped to underdamped, leading to justifications for the oscillations observed in the algorithm before settling to the minimum. However, the genesis of the second order ODE is unclear. In \cite{shi2022understanding}, high-resolution ODEs that incorporate higher-order terms, specifically those of order ($ {O(\sqrt s )}$ , with $s$ indicating the step size) to capture the nuanced dynamics that drive acceleration, such as momentum adjustments and gradient corrections was developed. The paper also derives the ODEs which differentiates between the Heavy-Ball method, NA-SC and the NA-C cases but these derivations of the ODE’s follow from the respective algorithms and thus fail to capture the underlying first principles behind the phenomena of acceleration and faster convergence. \\
Second-order continuous-time dissipative dynamical systems with viscous and Hessian driven damping \cite{alvarez2002second} have inspired effective first-order algorithms for solving convex optimization problems. While preserving the fast convergence properties of the Nesterov-type acceleration, it has been speculated that the Hessian driven damping makes it possible to significantly attenuate the oscillations, by building on the analogy of the accelerated schemes with the rolling of a Heavy-Ball down a slope. Polyak motivated momentum methods \cite{polyak1964some} by an analogy to a “Heavy-Ball” moving in a potential well defined by the quadratic cost function and provided an eigenvalue argument that his Heavy-Ball Method required no more iterations than the method of conjugate gradients \cite{wilson2021lyapunov}. Despite its intuitive elegance, Polyak’s eigenvalue analysis does not apply globally for general convex cost functions. In fact, Lessard et al. \cite{lessard2016analysis} derived a simple one-dimensional counterexample where the standard Heavy-Ball Method does not converge. Moreover in recent works \cite{goujaud2025open} it has been shown that the Heavy-Ball method provably does not reach an accelerated convergence rate on smooth strongly convex problems. \\
A recent work \cite{ross2022derivation} purports to derive the NA from first principles founded on optimal control theory, wherein the optimization problem is posed as an optimal control problem whose trajectories satisfy the necessary conditions for optimal control and generate various continuous-time algorithms. The necessary conditions produce a controllable dynamical system for accelerated optimization and stabilizing this system via a quadratic control Lyapunov function generates an ODE.  It is claimed that an Euler discretization of the resulting ODE produces Nesterov’s algorithm but this claim is hard to verify. \\
Nesterov’s proof techniques \cite{nesterov2013introductory} to accelerate convex optimization methods departed from the physical intuition previously used by Polyak, instead employing the method of estimate sequences to validate the effectiveness of momentum-based methods. However, the theoretical foundations of estimate sequences have remained elusive, with many researchers viewing the associated proofs as little more than an “algebraic trick.”\\
Several authors have proposed schemes to achieve acceleration without appealing to the estimation sequence \cite{bubeck2015geometric}, \cite{drori2014performance}, \cite{drusvyatskiy2018optimal}, \cite{lessard2016analysis}. 
A particularly promising direction in the study of acceleration involves analyzing the continuous-time limits of accelerated methods \cite{krichene2015accelerated}, \cite{su2016differential}, or deriving these limiting ODEs directly from an underlying Lagrangian framework \cite{wibisono2016variational}, and then establishing their stability using Lyapunov functions.  Despite these advances, such approaches do not provide a general principle for translating a continuous-time ODE into a discrete-time optimization algorithm.

\item Geometric interpretation: Other work has focused on geometric interpretations, showing that Nesterov's method can be understood as a process that strategically combines an ordinary gradient descent step with a momentum step. These interpretations often frame the algorithm in a way that unifies the gradient descent with mirror descent \cite{allen2014novel} to explain the phenomena of acceleration and avoid the original, obscure algebraic proofs, provided in the seminal work of Nesterov.
\end{enumerate}
In summary, Nesterov's method is fundamentally algebraic, relying on a set of precise mathematical relationships to achieve its optimal convergence rate. The "mystery" arises from the fact that these relationships are not intuitively obvious, and the algorithm itself doesn't behave like a traditional descent method. 
\subsection{Key Intuitions}
The primary focus of this work is to address the above mentioned gap, i.e. to provide insights into the faster convergence of the NA as well as  articulating the principles and the thought process behind the development of the algorithm. The NA is a first order gradient solver to minimize a function $f(x)$ and yet achieves the faster convergence rates of second order methods. The governing gradient flow equation is extended to the second order by posing the problem as a controlled dynamical system in the framework of the Passivity and Immersion (P\&I) paradigm \cite{nayyer2022towards}. This framework provides the evolution of the dynamical system on two time scales: slow dynamics evolving with a rate of $1/L$ on $M_0$ and the fast dynamics evolving on a manifold which is Normally Hyperbolic to $M_0$ at a time scale $\mu$ ($<< L$). \\
Constructing a Normally Hyperbolic Invariant Manifold (NHIM) \cite{wiggins2013normally} requires determining the Ehresmann Connection (EC) between the desired invariant manifold $M_0$ and the fiber bundle structure. With reference to the connection, the tangent space of the dynamics is resolved into Horizontal and Vertical components $(V_H$ and $V_V)$ which are orthogonal to each other, details of which are discussed in Section \ref{sec:3}. Yet, effectively the system evolves based on the time scale of the slower dynamics. In order to accelerate the system, a small and suitable perturbation dynamics $M_p$ (which decays exponentially to the equilibrium $x^*$) is added to the NHIM structure. Borrowing ideas from Geometrical Singular Perturbation theory, according to Fenichel's theorem, the Normally Attractive Invariant Manifold (NAIM) $M_0$ is preserved under perturbation and results in an invariant manifold $M_\epsilon$ that is diffeomorphic to $M_0$.\\
The effect of the perturbation $ M_p$ on the system leads to a fast inner layer dynamics which decays exponentially, and a slow perturbed outer layer. The interaction between the two layers improves the convergence of the ensuing dynamics. In essence, \textbf{mixing of the time scales causes a reduction in the spectral gap, leading to faster convergence}. 
Of importance to be highlighted here is that the above discussion gives a perspective to decipher the celebrated "estimation sequence" used in \cite{nesterov2013introductory}. The time scale of the perturbation dynamics decides the convergence rates of the algorithm and has to be carefully chosen to achieve optimal convergence, which are highlighted in Section \ref{sec:4}. We show how different time scale mixing after adding the perturbation results in the various schemes of optimal methods of NA. A section on extensions of the proposed idea to comprehend convergence in the convex case and the triple momentum algorithm is enumerated in \ref{sec:d}. The mathematical preliminaries and notations used to understand the work are discussed in Section \ref{sec:2}, with conclusions described in Section \ref{sec:5}. The Appendix \ref{sec:a} enumerates the connections between the time scale and step size and the foundations on which the P\&I approach is developed.    
\subsection{Contributions}
\begin{itemize}
\item A generalized approach to encompass the dynamics of acceleration of first order gradient methods to solve large scale unconstrained optimization problems has been proposed. 
\item The approach relies on the construction of a NHIM manifold which is suitably perturbed. This perspective gives insights into both the continuous time as well as the algorithmic evolution of NA, for both the strongly convex and the convex case.
\item Central to the approach is the concept of perturbing the NHIM, resulting in inner and the outer layer dynamics, and their ensuing interactions leading to a faster convergence of NA. The persistence of the perturbed invariant manifold is ensured by Fenichel's theorem, a central result in Geometric Singular Perturbation Theory. 
\item Viewed from this standpoint, the approach facilitates re-deriving the so-called ‘algebraic trick,’ which ceases to appear mysterious. A step-by-step derivation of the algorithm, grounded in the ideas of time-scale mixing and the preceding discussions, naturally leads to the formulation of the two-time-step algorithm proposed by Nesterov.
\item The proposed work also brings in ramifications regarding why the NA is not monotonous, reasons for non convergence in the Heavy-Ball algorithm and also can be extended to explain the acceleration obtained in the triple momentum method.
\end{itemize}
\section{Preliminaries and Notations} \label{sec:2}
\begin{itemize}
\item $f(x)$: function to be minimized, $x \in R^n$
\item $f(x^*)$: function value at the minima $x^*$
\item $\mu$: strong convexity parameter (i.e. $f$ is $\mu$ strongly convex), $\alpha=\mu$
\item $L$ is Lipschitz constant of the gradient of $f$, $\beta=1/L$
\item $M_0$, $M_\epsilon$, $M_p$, $M$: invariant manifold, the perturbed invariant manifold, the added perturbation to $M_0$, characterized by coordinates not on $M_0$ (nor on the pathway)
\item $\mathcal{O}$: Big O operator
\item $\kappa$: condition number of $f$
\item $R$: semi-Riemann metric
\item $w$: connection term
\item $\perp$: perpendicular
\item $TM_0, T_X$: tangent bundle of $M_0$, tangent bundle of the base space X of the fiber bundle
\item $V_H$, $V_V$: horizontal and vertical sub-bundles which together span the tangent spaces of $M_0$.
\item S: Storage function
\item $u,y$: control input, passive output.
\item $x_k, y_k, v_k,\phi_k, \alpha_k, \lambda_k, \beta_k$: iterate variables of the algorithm corresponding to $M_0$, inner and outer dynamics, estimation sequence, $\sqrt(\mu/L)$, weighting factor between 0 and 1, coupling coefficient.
\item Strong Convexity: A twice differentiable function $f:R^n\rightarrow R$ is $\mu $-strongly convex with $\mu>0$ being the modulus of convexity, if its Hessian matrix $\nabla ^2f(x) $ satisfies the matrix inequality $\nabla^2f(x) >\mu I\hspace{0.1cm} \forall x \in R^n$. It measures how steep the function is at its minimum.
\item L-Lipschitz Continuity: A differentiable function $f(x)$ is smooth iff it has a Lipschitz continuous gradient,
i.e., iff $\exists L < \infty$ such that $||\nabla f(x_1)-\nabla f(x_2)||_2 \leq L||x_1-x_2||_2, \forall x \in R^n $. It measures how fast the gradient changes between two points.
\item Condition number: For a function $f(x)$ that is L-smooth and $\mu$ strongly convex, the condition number is defined as $\kappa=\frac{L}{\mu}$. A large $\kappa$ indicates that the function is ill-conditioned (curvature varies significantly in all directions) while a small $\kappa$ indicates a well conditioned function (spherical bowl). For convex functions ($\mu=0$), the condition number is not well defined and convergence is analyzed in terms of L.
\item Rate of convergence: Describes how quickly a system approaches an equilibrium point or an attractor. It is defined as the rate at which $f(x_k)-f(x^*)$ decreases as $k\rightarrow \infty$. For a strongly convex function $f(x)$, gradient descent algorithm achieves convergence of the order of $\mathcal{O}((1-1/\kappa)^k)$, whereas the Nesterov's algorithm NA achieves convergence of the order of $\mathcal{O}((1-1/\sqrt(\kappa))^k)$.
\item Normally Attractive Invariant Manifold (NAIM): In dynamical systems theory a NAIM (a special case of NHIM ) is an invariant manifold with only stable normal directions, with the rate of contraction in the normal directions being stronger than any possible contraction or expansion within the manifold itself.
\item Fenichel's theorem : Let $M_0$ be a compact NHIM (attractive) of a dynamical system defined by a vector field $f$. For a perturbed vector field $f_\epsilon$, where $\epsilon$ is small, there exists a unique invariant manifold $M_\epsilon$ close to $M_0$ which is diffeomorphic to $M_0$ and preserves the normal hyperbolic (attractive) property.
\item Time scale and Step size: Step size in gradient descent algorithm (discrete domain) directly corresponds to time scales of the ODE (continuous domain). A small step size gives a fine grained approximation of the continuous dynamics. The connections between the step size and time scale gives insight into many optimization phenomena, the details are explained in the Appendix.
\item Why a second order dynamical system?:
For a L-smooth, $\mu$-strongly convex function $f$, the gradient descent time scale $L$ is given by : $ \dot x(t) = \frac{-1}{L}\nabla f(x(t)) $. This first order dynamical system converges to x* (the global minimum of $f$) at an exponential rate :
\begin{align*}
    f(x(t)) - f(x^*)= \{f(x(0)) - f(x^*)\}e^{-\frac{\mu}{L}t}\\= \{f(x(0)) - f(x^*)\}e^{-\frac{t}{\kappa}}
\end{align*}
This indicates that larger the condition number, slower the convergence. The gradient descent achieves a convergence rate of: 
$\mathcal{O}((1-1/\kappa)^k)$. From a dynamical theory perspective the only way to improve the condition number (so as to speed up convergence) is to introduce an additional time scale $\mu$ and increase the order of the system. This forms the basis for mixing of time scales discussed in Section \ref{sec:3}, resulting in improvement of convergence rate (indirectly through $\kappa$) as suggested by Nesterov.
\end{itemize}
\section{Passivity and Immersion $(P\& I)$ method} \label{sec:3}
This section describes the $P\&I$ formulation to arrive at the continuous time counterpart of the NA. In essence, the $(P\& I)$ method \cite{nayyer2022towards} reframes the minimization of a function $f(x)$ into the stabilization of the manifold formed by the desired dynamics (the gradient flow). The stability of the manifold through the design procedure ensures that the objective function $f(x) $ attains a minimum. The details and genesis of the algorithm are explained in the Appendix.
The unconstrained optimization problem is defined as: 
\begin{align}\label{eq:21}
  \min_{x \in \mathbb{R}^n} f(x)  \\ f(x) \in  S^{1,1}_{\mu,L} \nonumber 
\end{align}
 The governing gradient flow
equation is extended to the second order by posing the problem as a controlled dynamical system:
 \begin{align}\label{eq:22}
 \begin{gathered}
  \dot x_{1} = x_{2} \\
  \dot x_{2} = u  
  \end{gathered}
 \end{align}
With $x \in (x_{1},x_{2}) \in R^{2}$, we consider a two dimensional state space for the ease of presentation and it can be extended to $x \in R^{2n}$.
 Since the gradient flow is the desired dynamics needed to solve the optimization problem in the given state space, choose a manifold $M_{0}$ (shown in Fig \ref{fig:fig1}) defined as:
 \begin{align}\label{eq:23}
 M_{0}  = \{(x_{1},x_{2})\mid x_{2} + \beta \nabla f(x_{1})=0\}
 \end{align}
 Note : From (\ref{eq:23}) one can infer that the dynamics on $M_0$ is of the form: $ \dot x_{1} = - \beta \nabla f(x_{1}) $
 \begin{lemma}
Given the form of $M_{0}$ as in (\ref{eq:23}),it can be proved that $M_{0}$ is an invariant manifold, invariant under the gradient flow $x_{t}$ defined by the equation:  $\dot x_{t} = -\beta  \nabla f(x(t)),\hspace{.1cm} \forall t \in R$ and $x(t) \in M_{0}$.
\end{lemma}
 Thus, $M_{0}$ is an invariant manifold in the given state space that contains the equilibrium point $(x_{1}^{*},0)$ of the gradient flow.\\
 In general, consider a manifold $M$ which is obtained by integrating the dynamics on the manifold, i.e.
 $\int \dot M dt=\int (\dot x_2+b \nabla ^2 f(x_1) \dot x_1) dt$. This is an indefinite integral resulting in:
 $M=x_2+ \beta \nabla f(x_{1})+c=0$. When $c=0, M=M_0$, else $M$ represents coordinates not on $M_0$.\\
 We briefly discuss the steps involved in the P\&I approach involving the splitting of tangent space and the control law design: 
 \begin{enumerate}
     \item Define the target dynamics of the system (\ref{eq:22}) and hence the desired manifold $M_0$.
     \item Calculate the EC, $R$ and the connection term $\omega$ on the fiber bundle of (\ref{eq:22}).
     \item Decompose the vector field (\ref{eq:22}) into $V_H$ and $V_V$ with respect to $\omega$, with the desired dynamics along $V_H$.
     \item Design the control law $u$ (lying along $V_V$) so as to move the off-manifold trajectories towards the manifold $M_0$.
 \end{enumerate}
 Detailed explanations of these steps appear in the following sections.
\begin{figure}
    \centering
    \includegraphics[width=0.5\linewidth]{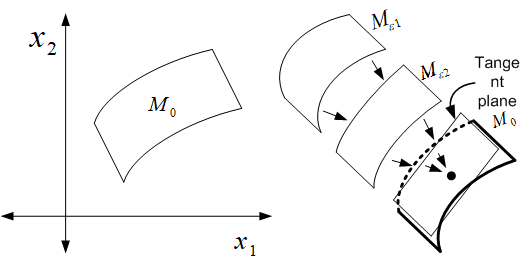}
    \caption{Invariant manifold $M_0$ and the effect of perturbation on $M_0$}
    \label{fig:fig1}
\end{figure}

 \subsection{Splitting of the Tangent space}
 \begin{itemize}
     \item Given $M_{0}$, the normal to $M_{o}$ is calculated as $M_{o}^{\perp} = [\beta  \nabla^{2}f(x_{1} ) \hspace{.2cm} 1]$
     \item The semi-Riemann metric $R$ and the connection $\omega$ is derived as:\\
     $R=(M_{o}^{\perp})^T M_{o}^{\perp}$
     \item 
    $$ R =\begin{bmatrix}
     (\beta \nabla^{2}f(x_{1}))^{2} & \beta \nabla^{2}f(x_{1})\\
     \beta \nabla^{2}f(x_{1}) & 1
 \end{bmatrix}
 =\begin{bmatrix}
     M_{11} & M_{12}\\M_{21} & M_{22}
 \end{bmatrix}$$
 \item The connection term $\omega = \frac{M_{21}}{M_{22}} = \beta \nabla^{2}f(x_{1})$\
 \item For the given EC, $\omega$, the tangent space of $X$ splits as a direct sum of the horizontal space $V_H$ and vertical space $V_V$ as shown below:
    $$(\dot x_1, \dot x_2) = V_H \oplus V_V $$ =
    $(\dot x_1,-\beta \nabla^2 f(x_1){\dot x_1}) \oplus (0,\dot x_2 +  \beta \nabla^2f(x_1)\dot x_1)$
     \end{itemize}
\begin{lemma}
$V_H$ is tangential to $M_0$
\end{lemma}
\begin{proof}
$M_{o}^{\perp}$ is defined as 
    $ \begin{bmatrix}
        \beta \nabla^2f(x_1) & 1 
    \end{bmatrix}$  \\
    Along $V_H$, $\dot x_2 = -\beta \nabla ^2 f(x_{1})\dot x_1$\\
    $\therefore M_0^{\perp}.V_H = \begin{bmatrix}
     -\beta \nabla ^2 f(x_{1}) & 1  
    \end{bmatrix}
    \begin{bmatrix}
        \dot x_1\\
         -\beta \nabla ^2 f(x_{1})\dot x_1
    \end{bmatrix}=0$
        \end{proof}
Now to ensure that $M_0$ is attractive i.e. all off manifold trajectories in $X$ converge to $M_0$ exponentially a control law $ u$ is derived as follows.     
\subsection{Finding the Control law}
 \begin{itemize}
\item   Define $S = \frac{1}{2}M^2$ where $M =  {x_2  +\beta \nabla f(x_1) \neq 0}$. The definition of $M$ implies that the system trajectories are not on $M_0$.
\item Now, $u$ is chosen so that $\dot S \leq -\hat \alpha S$, i.e 
\begin{align}\label{eq:28}
    M\dot M\leq-\hat \alpha/2 M^2\\
    \dot M \leq \frac{-\hat \alpha}{2} M \nonumber
\end{align}
One has :   $u + \beta  \nabla^2 f(x_1)\dot x_1 \leq \frac{-\hat \alpha M}{2}$ \vspace{.1cm} \\
$u = -\beta \nabla^2f(x_1)\dot x_1 - \alpha M \hspace{0.5cm}$ (choosing $\alpha=\frac{\hat \alpha}{2})$
   \item  Substituting $u$ in (\ref{eq:22})
   \begin{align}
       \dot x_1 = x_2 \nonumber \\
     \dot x_2 = -\beta \nabla ^2 f(x_{1})\dot x_1 - \alpha(x_2 + \beta\nabla f(x_1))\nonumber  \\
    =-\beta \nabla ^2 f(x_{1})\dot x_1-\alpha x_2-\beta\alpha \nabla f(x_1) \nonumber
   \end{align}
    \end{itemize}
\textbf{Remark I}: Since $M = \{(x_1,x_2) | x_2 + \beta \nabla f(x_{1} \neq 0\})$, $\dot M = V_V =     (0,\dot x_2 + \frac{1}{L} \nabla ^2 f(x_{1})\dot x_1)$ and $\dot M \leq -\alpha M \implies \dot V_V \leq -\alpha V_V $. Thus, the controlled trajectories are along the $V_V$ direction and thus transverse (orthogonal) to the manifold $M_0$. 
\vspace{.1cm}\\
\textbf{Remark II}: Given $M_0$, $M$ can be viewed as the manifold $ M_\zeta = \{(x_1,x_2) | x_2 +\beta \nabla f(x_{1}) = \zeta \}$. 
Under the control action along the $V_V$ direction, as $(\zeta \rightarrow 0), M_\zeta \rightarrow M_0 $, as shown in Fig \ref{fig:fig1}. \vspace{.1cm} \\
Thus under the action of the control $u$ in (\ref{eq:22}) the manifold $M_0$ becomes attractive and any off-manifold trajectory converges exponentially to the equilibrium $(x_1^*,0)$ contained in the invariant manifold $M_0$. \\
Thus for the dynamical system $(\mu >> \frac{1}{L})$:
\begin{align}\label{eq:25}
\begin{gathered}
  \dot x_1=x_2\hspace{.2cm}   \\
    \dot x_2 = \beta \nabla ^2 f(x_1)\dot x_1 - \alpha x_2 - \beta \alpha  \nabla f(x_1)
    \end{gathered}
\end{align}
The manifold $M_0$ is a NAIM, with the rate of contraction along $M_0$ being $\beta$ and the rate of contraction in a direction orthogonal to $M_0$ being $\alpha$.\\
If one designates $x_1=x$ then (\ref{eq:25}) can be written as :
\begin{align}\label{eq:26}
 \ddot x + \beta \nabla^2 f(x_1)\dot x +  \alpha \beta \nabla f(x) = 0   
\end{align}
\section{Addition of Perturbation dynamics and Time scale mixing}\label{sec:4}
 The NAIM $M_0$ is characterized by the existence of a large spectral gap between the slow dynamics on $M_0$ and the fast transverse dynamics, implying that the time scale of the slow dynamics ($\frac{1}{L}$) would govern the rate of convergence. For accelerated convergence, it is essential to bridge the spectral gap, achieved by employing time-scale mixing, which refers to scenarios wherein the slow-fast separation breaks down (as in when fast and slow scales interact due to parameter values \cite{eilertsen2020quasi}).  Time scale mixing is achieved through the addition of a suitable (decaying) perturbation  affecting the NAIM structure. However, the NAIM structure of $M_0$ is preserved during the perturbation, due to the applicability of the Geometric Singular Perturbation Theory (GSPT) \cite{fenichel1979geometric} and Fenichel's theorem. \vspace{.1cm} \\
 What constitutes as a perturbation of $M_0$ - any input to the present framework which does not evolve along the pathway (defined by the connection). Hence a suitably defined  manifold $M_p$ is added to the NAIM structure defined by(\ref{eq:25}), and the resultant system dynamics follows from (\ref{eq:28}),  i.e\\
 \begin{equation}\label{eq:27}
     \dot M \leq -\alpha(M + M_p)
 \end{equation}
    The manifold $M_p$ is of the form $M_p= \{(x_1,x_2)\mid \alpha x_2 + \nabla f(x_1) = \eta \}, \eta\rightarrow 0$, implying that $M_p$ evolves with time to the tangent space of $M_0$, i.e $M_p \rightarrow TM_{0x^*}; x^*=(x_1^*,0)$.\\
    Note: From the above it is clear that $M_p$ is a decaying manifold. With the addition of $M_p$, equation (\ref{eq:26}) is modified as:
    \begin{align}\label{eq:28a}
        \ddot x + \beta \nabla^2 f(x_1)\dot x + \alpha [(\dot x +\beta \nabla f(x))+(\dot x +\frac{1}{\alpha} \nabla f(x))] = 0
    \end{align}
 
    \textbf{Remark III:} Equation (\ref{eq:25}) can be viewed as a linear system defined as:
    $$ \dot x_1=x_2$$
    $$ \dot x_2=u_1$$ where $$ u_1= -\beta \nabla^2 f(x_1)x_2 - \alpha x_2 -  \alpha \beta \nabla f(x_1)$$
    The equilibrium point is given as $(x_1^*,0)$ as discussed. The addition of the perturbation can be viewed as the addition of an input $u_2$ to the above equation leading to:
    $$ \dot x_1=x_2$$
    $$ \dot x_2=u_1+u_2$$ where $$ u_2= -\alpha x_2- \nabla f(x_1) $$ leading to (\ref{eq:28a}).
Since the above system is linear, to understand the effect of $u_2$, one takes $u_1=0$. This leads to:
$$\dot x_1=x_2 $$
$$\dot x_2=u_2=-\alpha x_2-\nabla f(x_1) $$   

So as to study the effect of $u_2$, put $\dot x_1=\dot x_2=0$ entailing $x_2=\dot x_1=-(1/\alpha) \nabla f(x_1) $. Since $f(x_1)$ is a smoothly convex function: $\dot x_1=-(1/\alpha) \nabla f(x_1) \implies x_1 \rightarrow x_1^*$. Thus the equilibrium point $(x_1^*,0)$ given by the condition $x_1^*=x_2^*=0$ is asymptotically stable, and $u_2$ qualifies to be a decaying perturbation. A more formal proof can be found in \cite{gunjal2023new}.
\subsection{Transient Time scale mixing}
Having defined what constitutes a perturbation $M_p$, we now investigate its effect on $M_0$. The manifold $M_0$ is characterized by the slow-fast NAIM structure. Due to the perturbation $M_p$, the spectral separation is temporarily lost, causing the fast normal attraction to slow down and interfere with the slow motion along $M_0$. This occurs when the perturbation decay rate aligns with one of the system's inherent time scales, thereby blurring the slow-fast distinction leading to transient non-hyperbolic behaviour of the system dynamics. Importantly, while the mixing phenomena is transient, resolving as the perturbation decays to zero and asymptotic separation is restored, the NAIM may experience temporary loss of attractiveness or smoothness during this phase. This manifests in the emergence of an intermediate time scale, which one refers to as the process of time scale mixing. As the perturbation decays, the system reverts asymptotically to the unperturbed- slow fast dynamics, restoring the NAIMs attractiveness. This transient mixing is a tunable phenomena offering insights into various real world systems with varying disturbances, particularly in the field of Biology, Neuroscience, and Chemistry, to name a few.\\
\textbf{Remark IV:} In the above process, the structure of NAIM and Fenichel's theorem play a central role, else, even for a decaying perturbation, the invariant manifold may fail to recover and probably degenerate to an invariant set (a probable reason behind the non-convergence and appearance of cycles for certain regimes in case of the Heavy-Ball dynamics \cite{goujaud2025open}.
\subsection{The structure of the perturbed NAIM}
To describe a perturbed NAIM, one often uses the concept of inner and outer layer dynamics, to approximate the behaviour of the system dynamics near the perturbed manifold. 
In the context of dynamical systems, a perturbed normally invariant manifold $M_\epsilon$ consists of :
\begin{itemize}
    \item Inner layers : characterized by stable dynamics perpendicular to the manifold surface. They form the basin of attraction surrounding the manifold and exponentially converge towards the manifold.
    \item Outer layers: comprises of the invariant manifold itself and is characterized by internal dynamics that persist after perturbations.  When perturbed, NAIM maintains structural stability. Small perturbations preserve the manifold’s existence and basic topological properties although its exact position and internal dynamics may slightly shift.
    \item As shown in Fig \ref{fig:f2}, the outer layer is tangential and diffeomorphic to $M_0$ while the inner layer converges exponentially.
\end{itemize}
\begin{figure}[h!]
    \centering
    \includegraphics[scale=0.2]{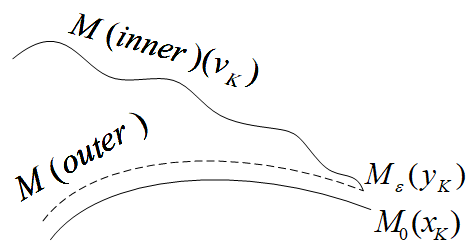}
    \caption{The inner and outer layer manifolds after perturbation with the associated variables}
    \label{fig:f2}
\end{figure}
Therefore, in NAIM, the inner layer exhibits faster dynamics and the outer layer exhibits slower dynamics constrained to the manifold surface. With the mixing of time scales of these dynamics, various resultant rates can be achieved. The geometric mean arises when the multiplicative interactions dominate and the effective time scales often combine multiplicatively rather than additively. 
Owing to the mixing, the traditional separation between inner layer and the  outer layer breaks down and the convergence rates reach similar values. Therefore, the inner and outer layer dynamics of a perturbed NAIM can indeed operate at the same rate, having a time scale that is proportional to geometric mean $\sqrt{\alpha/\beta}$. This rate combines multiplicatively with the slow layer $\frac{1}{\alpha}$ resulting in an effective time scale $\sqrt{\frac{1}{\alpha \beta}}$ for the overall dynamics.
To sum up, the equivalent gradient flow dynamics of the combined layer (with rate $\sqrt \frac{\beta}{\alpha}$) is given as: $$ \dot x = -{\sqrt \frac{\beta}{\alpha} }\nabla f(x) $$
Given that f(x) is $\mu (or \hspace{.1cm}\alpha)$-strongly convex and $L$ Lipschitz continuous, the effective rate of convergence is proportional to $\alpha{\sqrt{\beta/\alpha }} = \sqrt{\alpha \beta} = \sqrt{\frac{\mu}{L}} =\frac{1}{\sqrt \kappa}$ as desired. \\
Finally to reflect the time scale mixing, the parameters of (\ref{eq:28a}) are modified as: $\beta \rightarrow \sqrt \beta= \sqrt{\frac{1}{L}}, \alpha \rightarrow \sqrt \alpha= \sqrt{\mu} $. 
This leads to the final form:
\begin{align}
    \ddot x + \sqrt{\frac{1}{L}} \nabla^2 f(x_1)\dot x + 2\sqrt{\mu}\dot x + (1+\sqrt{\frac{\mu}{L}})\nabla f(x) = 0 \nonumber
\end{align}
The form of the above equation matches exactly with the high resolution ODE \cite{shi2022understanding} for L-smooth and $\mu$-strongly convex functions.
Based on the ideas obtained from the above analysis (the role of perturbation, inner and outer layer dynamics, geometric time scale mixing), section \ref{sec:4} is developed to interpret and derive the sequence estimation method of Nesterov, assuming that the time scale for the overall dynamics is $\sqrt{\frac{\mu}{L}}$. 
\subsection{Issues with Heavy-Ball dynamics}
This section provides a strong justification as to why the Heavy-Ball dynamics oscillate \cite{polyak1964some} and fail to converge in certain cases for strongly convex functions.
\begin{itemize}
\item As discussed in the previous section, the dynamics of acceleration in the Nesterov's flow is due to the perturbation of the NHIM and the resultant time scale mixing.
    \item The above is contingent on the structural stability of '$M_0$', (the invariant manifold $M_0$ (NAIM)) and this is ensured by the application of Fenichel's theorem.
    \item The NAIM framework plays a central role in the development of the above results because of the path defined in the state space '$X$ by the connection term, i.e. \textbf{the Hessian of $f$}.
    \item The tangent bundle '$T_X$' is decomposed using the direct sum as:
    \begin{equation}
        T_X = V_H \oplus V_V = T {M_0} \oplus V_V
    \end{equation}
    \item This decomposition results in the NAIM framework that ensures the structural stability of $M_0$.
    \item The Heavy-Ball method mimics the Nesterov's flow but \textbf{without the Hessian term}. Thus, when perturbation is introduced, the manifold $M_0$ can easily loose its properties (both topological and dynamical) and the resultant dynamics after the addition of the perturbation may cause the system to be non-convergent to the optimal $x^*$ \cite{goujaud2023provable} .
\end{itemize}
\section{Nesterov's Algorithm decoded}\label{sec:4}
The major contribution of this section is to explain the physical intuition behind the NA \cite{nesterov2013introductory} which, in literature, has been deemed to be an "algebraic trick". With supporting rationale from the P \&I  approach and time scale mixing, described in the earlier sections, the idea behind the Nesterov's estimation sequence and the ensuing expressions that lead to the celebrated NA are decoded. The parallels between the proposed method (Section \ref{sec:3} and \ref{sec:4}) and the NA is discussed in this section. \\
\subsection{Background}
Given a function $f \epsilon S_{\mu,L}^{1,1}(R^n)$, to be minimized, effectively implying that the solution to the gradient flow dynamics has to be found. Formulating this as a control problem:
$$\dot x_1=x_2$$
$$\dot x_2=u$$ and applying the P\&I approach results in splitting the tangent space of the desired dynamics into $V_H$ and $V_V$, leading to a NHIM structure. However, the large time scale separation between $V_H$ and $V_V$ eventually causes the second order system to behave as a first order system. Adding a perturbation $M_p$ (whose dynamics are not along $V_V$) causes a perturbation of $M_0$, thereby reducing the spectral gap between $V_H$ and $V_V$. The dynamics of the perturbed manifold $ M\epsilon$ can be analyzed by resolving the effect of the added perturbation into inner and outer layer dynamics, with the latter lying tangential to $ M\epsilon$ and slowly decaying, finally resulting in $M_{\epsilon} \rightarrow TM_{0}\mid(x_{1}^{*},0)$ as $k\rightarrow \infty$.\\
Reviewing the above discussion:
 \begin{itemize}
     \item $M_{0}$ - NHIM defined as $M_{0} = \{(x_{1},x_{2})\epsilon R^{2n}|x_{2} + \frac{1}{L}\nabla f(x_{1}) = 0\}$
     \item $M_{P}$ - defines the perturbation of $M_{0} : M_{p} = \{(x_{1},x_{2})\epsilon R^{2n}|x_{2} + \frac{1}{\mu}\nabla f(x_{1})\}$
     \item $M_{\epsilon}$ - perturbed manifold characterized by inner and outer layer dynamics. Note that $M_{\epsilon}$ is diffeomorphic to $M_{0}$.
 \end{itemize}
Based on the perspective gained on the dynamics of acceleration for the gradient flow, an iterative process is designed to mimic the behaviour of the accelerated gradient flow achieved through the process of time scale mixing as explained.\\
In the iterative process at each instant $k$ we assign coordinates/variables as follows :
\begin{itemize}
    \item $v_{k} (\in R^n) \rightarrow$ iterate variable representing inner dynamics
\item $y_{k}(\in R^n) \rightarrow$ iterate variable representing outer dynamics
\item $x_{k}(\in R^n) \rightarrow $ manifold $M_{0}/M_{\epsilon}$ which is diffeomorphic to $M_{\epsilon}/M_0$
\end{itemize}
Also, $M_{\epsilon} \rightarrow TM_{0}|x_{1}^{*}$ as $k\rightarrow \infty$ and $v_{k} \rightarrow x_{k}$. 
The succeeding section highlights the algorithmic viewpoint of the above discussion and the evolution of the estimation sequence given by Nesterov.
\subsection{NA- decoded}

We now interpret Nesterov's accelerated convergence algorithm for strongly convex functions (NA-SC). \vspace{.2cm} \\
A pair of sequences $\{ \phi_{k}(x)\}_{k=0}^\infty$ and $\{\lambda_k\}_{k=0}^\infty : k \geq 0$, 
called the estimation sequence is defined (as shown in Fig. 3) i.e. 
\begin{align}\nonumber
 \phi_{k}(x) \leq (1-\lambda_k)f(x) + \lambda_{k}\phi_{0}(x)  
\end{align}
The sequence $\phi_{k}(x)$ models the evolution of $M_{\epsilon}$, so if for some sequence $\{x_{k}\}$: 
\begin{align} \nonumber
    f(x_{k}) \leq \phi_{k} \equiv min_{x\epsilon R^n} \phi_{k}(x)
\end{align}
then $f(x_k) \rightarrow f^*$.
\begin{figure}[h]
        \centering
        \includegraphics[scale=0.2]{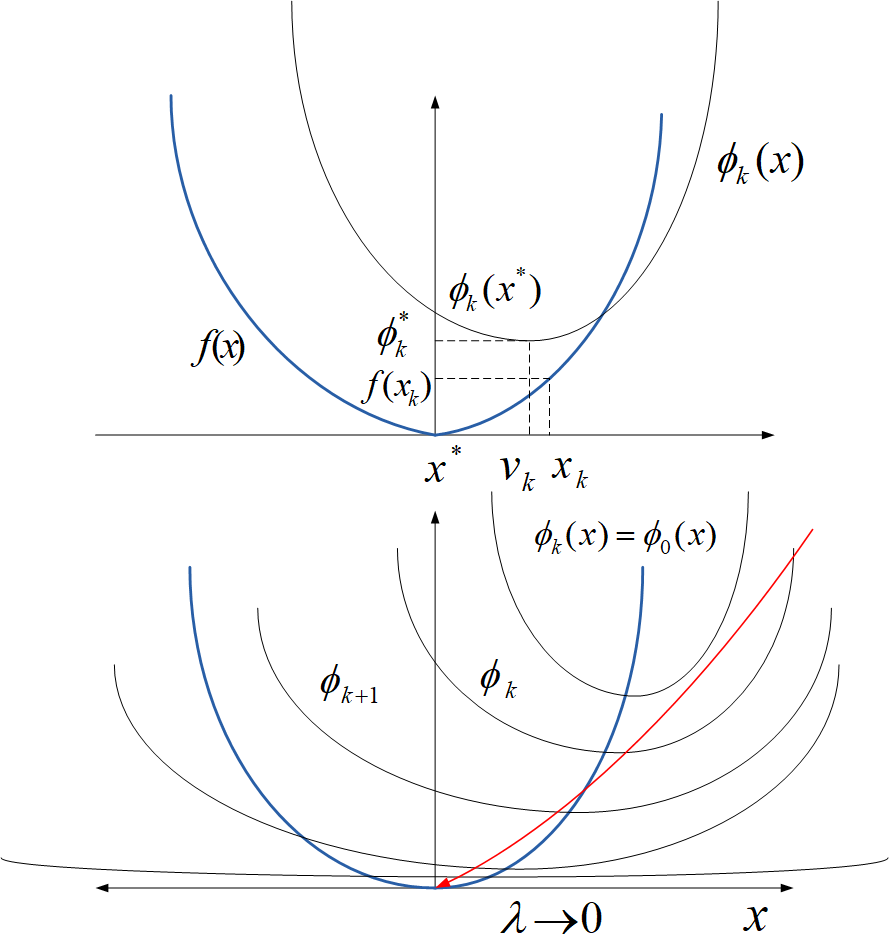}
        \caption{Estimation sequence $\phi_k(x)$ and its evolution}
        \label{fig:f1}
        \end{figure}
Let $\phi_{0}(x)$ be an arbitrary function of $R^{n}$. To tie up the idea of time scales and the associated weights into the iterative algorithm, a sequence $\{\alpha_k\}_{k=1}^{\infty}$ is introduced where : \\
$\alpha_k \epsilon (0,1) : \sum_{k=0} ^ \infty\ \alpha_k = \infty $ along with $\lambda _0 = 1$.\\
(In the case considered in this paper $\alpha_k = \alpha = \sqrt{\frac{\mu}{L}} \hspace{.2cm}\forall k$), then with $\lambda_{k+1} = (1-\alpha_k)\lambda_k$ one has:\\
\begin{align}\label{eq:a33}
\phi_{k+1}(x)=\underbrace{ (1-\alpha_k)\phi_k(x)}_I+
\underbrace{\alpha_k[f(y_k)+\langle f'(y_k),(x-y_k) \rangle+\frac{\mu}{2} \hspace{0.2cm}||x-y_k||^2]}_{II}
\end{align}
\textbf{Explanation:} The term $I$ defines the autonomous dynamics of $\phi_k$. Since $\alpha_k <1$ it implies that the dynamics is stable and that $\phi_k(x) \rightarrow TM_{O}|x^*$ as $k \rightarrow \infty$.\\
The term II acts as an input, representing the Taylor's series expansion of $f$ at $y_k$ (ties up the coordinate $x_{k}$ of $M_{0}$ with the outer layer coordinates $y_k$), where the outer layer is tangent to $M_0$. $\alpha_k$ gives the weights associated with the inner layer and the outer layer respectively as discussed. \\
\textbf{Note}: We take $\Gamma_k = \mu \hspace{.2cm}\forall k$.\\
The dynamics of $\phi_k(x)$ as defined in (\ref{eq:a33}) preserves the canonical form of the function defined in following equation
\begin{align}\nonumber
  \phi_{k}(x) = \phi_{k} ^ * + \frac{\mu}{2}||x-v_k||^2 \hspace{.1cm}\forall k : \Gamma_o = \Gamma_k=\mu 
\end{align}
Since $M_{\epsilon} \rightarrow TM_{0}|x$ as $k \rightarrow \infty$, the term $||x-v_k||^2$ ties up the inner layer dynamics (coordinate $v_k$) with the dynamics on $M_{0}$ (coordinate $x_{k}$).\\
This is proved by the following inductive-deductive argument:\\ \vspace{.2cm}
Assume $\phi_k^{\prime\prime}(x) = \mu I_n$ then, \\ $\phi_{k+1}^{\prime\prime}(x) = (1-\alpha)\phi_k^{\prime\prime}(x) + \alpha \mu I_n (from (\ref{eq:a33}))$\\
$\therefore \phi_{k+1}^{\prime\prime}(x) = \mu I_n$\\
Hence, $\phi_{k+1}(x) = \phi_{k+1}^* + \frac{\mu}{2}||x-v_{k+1}||^2 $\\
The first order optimality condition $\phi_{k+1}^{'}(x) = 0$ (when the inner layer tends to $M_{\epsilon}$ ) gives the iterate defining the evolution of $v_k$ i.e. 
\begin{align}
    v_{k+1} = \frac{1}{\mu}[(1-\alpha)\mu v_k + \alpha \mu y_k - \alpha f^{'}(y_k)]\nonumber \\
   = (1-\alpha)v_k + \alpha y_k - \frac{\alpha}{\mu}f^{'}(y_k) \nonumber
\end{align}
The autonomous dynamics of $v_k$ i.e. $v_{k+1} = (1-\alpha)v_{k} $ is stable, since $\alpha = \sqrt{\frac{\mu}{L}} < 1$\\
The forcing terms are as follows:
\begin{align}\nonumber
    \sqrt{\frac{\mu}{L}}y_k \hspace{.2cm} and \frac{1}{\sqrt{\mu L}} \nabla f(y_k)
\end{align}
The outer layer coordinate $y_k$ is weighted by $\sqrt{\frac{\mu}{L}}$ and the term $f^ \prime (y_k)$ by $\frac{1}{\sqrt{\mu L}}$ which agrees with the conclusions drawn from Section 4.2. \\
Finally, one sets up the computation of $\phi_{k+1}^{*}$ followed by deriving the conditions that ensure that 
$\phi_{k+1}^{*}(x) \geq f(x_{k+1})$ assuming $\phi_{k}^{*}(x) \geq f(x_{k})$\\
The conditions under which $\phi_{k+1}^{*}(x) \geq f(x_{k+1})$ result in the NA as shown below: \\
\underline{One starts by computing $\phi_{k+1}^{*}(x)$}
\begin{align}\label{eq:a36}
    \phi_{k+1}(y_k) = \phi_{k+1}^{*} + \frac{\mu}{2}||y_k - v_{k+1}||^2
\end{align}
Also from (\ref{eq:a33}) one has: 
\begin{align}\label{eq:a37}
    \phi_{k+1}(y_k) = (1-\alpha)[\phi_k^* + \frac{\mu}{2}||y_k -v_k||^2) + \alpha_kf(y_k)]
\end{align}

From (\ref{eq:a36}) and (\ref{eq:a37}) one has:
\begin{align}\label{eq:a38}
    \phi_{k+1}^{*} + \frac{\mu}{2}||y_k - v_{k+1}||^2 = (1-\alpha)[\phi_k^* + \frac{\mu}{2}||y_k-v_k||^2) + \alpha_k(f(y_k)]
\end{align}

On replacing $||y_k - v_{k+1}||$ in terms of $||y_k - v_{k}||^2$ and other terms using (\ref{eq:a38}), one gets from (\ref{eq:a36}, \ref{eq:a37},\ref{eq:a38}) and the assumption that $\phi_k^* \geq f(x_k)$ the following inequality:
\begin{align}\label{eq:a39}
    \phi_{k+1}^* \geq (1-\alpha)f(x_k) + \alpha f(y_k) - \frac{\alpha^2}{2\mu}||f^{\prime}(y_k)||^2 + \nonumber \\ \frac{\alpha(1-\alpha)\mu}{\mu}\langle f^{\prime}(y_k),v_k-y_k\rangle
\end{align}
Since $f(x_k) \geq f(y_k) + \langle f^{'}(y_k), x_k - y_k\rangle $ along with the equality:
\begin{align}\nonumber
    x_{k+1} = y_k - \frac{1}{L}f^{\prime}(y_k)
\end{align}
One has from (\ref{eq:a39}) the following inequality:
\begin{align}\nonumber
    \phi_{k+1}^{*} \geq (1-\alpha)f(x_{k+1}) + (1-\alpha)\langle f^{'}(y_k), \alpha(v_k - y_k) + x_k - y_k \rangle
\end{align}
Setting $\alpha((v_k-y_k) + x_k -y_k) = 0$ one has: 
$\phi_{k+1}^* \geq f(x_{k+1})$ subject to:
\begin{align}\nonumber
    \alpha(v_k - y_k) + (x_k - y_k) = 0 i.e. y_k = \frac{\alpha v_k + x_k}{1+\alpha}
\end{align}

In the above equation $v_k$ is substituted in terms of the manifold $M_0$ coordinates i.e.\\
$$v_{k+1} = x_{k} + \frac{1}{\alpha}(x_{k+1} - x_{k})$$
Since $y_{k+1} = \frac{1}{1+\alpha}[\alpha v_{k+1} + x_{k+1}]$ one has that:
\begin{align}\nonumber
    y_{k+1} = x_{k+1} + \frac{1}{\beta_k} (x_{k+1} -x_k)
\end{align}
  where
  \begin{align}\nonumber
      \beta_k = \frac{\sqrt L - \sqrt \mu}{\sqrt L + \sqrt \mu}
  \end{align}
Finally one has the following accelerated algorithm : 
\begin{align}\label{eq:a315}
    x_{k+1} = y_k - \frac{1}{L}f^{\prime}(y_k)
\end{align}
\begin{align}\nonumber
        y_{k+1} = x_{k+1} - \beta_k(x_{k+1} - x_k) 
\end{align}
where $\beta_k = \frac{\sqrt L - \sqrt \mu}{\sqrt L + \sqrt \mu}$\\
As can be seen, the gradient step (with step size $\frac{1}{L}$) taken by projecting the slow outer layer dynamics, tangent to $M_0$, onto $M_0$ as in (\ref{eq:a315}) and then updating the outer layer coordinate from $y_k$ to $y_{k+1}$ by using the coefficient $\beta_k$ which defines the coupling between the outer layer dynamics and the inner layer dynamics. The next gradient step then proceeds with $\phi_{k+1}$. \\
\textbf{Note}: The above outer, inner layer formation is possible because of the NAIM $M_0$. This structure arises because of the Hessian Term $\frac{1}{L} \nabla^2 f(x)\dot x$ as explained. In the case of the Heavy-Ball algorithm especially for the strongly convex case the absence of the Hessian term implies that the NAIM structure does not exist and the added perturbation makes the invariant manifold $M_0$ structurally unstable.
\section{Discussions}\label{sec:d}
The dynamics of acceleration of any first order algorithm rests upon the process of time scale mixing, achieved by the introduction of a suitable decaying perturbation to a NAIM, encoding a slow-fast dynamical system. The structural stability of the perturbed dynamics follows from Fenichel's theorem. With the addition of the perturbation, the dynamics of the gradient descent on the invariant manifold $M_0$ (coordinates $x_k$) is influenced by the outer (coordinates $y_k$) and inner (coordinates $v_k$) layer dynamics approximating the perturbed manifold, with the outer layer dynamics being tangential to $M_0$. Using ideas discussed in section 4, and the estimation sequence $\phi_k$ to denote the notion of perturbation and invoking the ideas of smoothness and strong convexity associated with the objective function $f(x)$, the accelerated gradient descent algorithm is arrived at as a two step algorithm involving $x_k$ ad $y_k$ only, after the substitution of $v_k$ related terms in terms of $x_k$ and $y_k$, leading to the accelerated algorithms:
\begin{align}\label{eq:d1}
    x_{k+1} = y_k - \frac{1}{L}f^{\prime}(y_k)
\end{align}
\begin{align}\label{eq:d2}
        y_{k+1} = x_{k+1} - \beta_k(x_{k+1} - x_k) 
\end{align}
where $\beta_k = \frac{\sqrt L - \sqrt \mu}{\sqrt L + \sqrt \mu}$
The above besides the parameters $(L and \beta_k)$, is the general structure associated with (any) first order optimization algorithms for a given strongly convex function. The hyper-parameters associated with different variants of (\ref{eq:d1}, \ref{eq:d2}) can be found using automated methods such as Performance Estimation Problem (PEP) framework \cite{colla2023automatic}, or control theoretic methods. The latter provide a powerful lens for analysis and automating parameter selection by treating algorithms as dynamical systems with feedback loops \cite{van2025fastest}. By modelling the optimization process as a closed loop interconnection between a linear time-invariant system (representing the algorithm's updates)and a nonlinear uncertainty (e.g. the gradient oracle satisfying sector bounds or convexity assumptions), these approaches leverage concepts like dissipativity (state-space based) and Integral Quadratic Constraints (IQC) input-output formulation to derive optimal parameters, convergence rates and stability parameters via solvable optimization problems like semi-definite programs (SDPs) or linear matrix inequalities (LMIs). This enables automated tuning without manual guessing, often achieving near-optimal rates for the accelerated algorithms. Control-theoretic methods excel by providing provable guarantees and automation, especially for structural problems.\\
\subsection{The smooth convex case}
If a function $f$ is $\mu$ strongly convex, then the function $x\longmapsto f(x)-\mu/2||x||_2^2$ is convex. Therefore a convex function can be viewed as one that is strongly convex with a parameter $\mu=0$. Thus accelerated dynamics (12) can be adapted to the convex case by studying its behaviour as $\mu \rightarrow 0$. Thus if (12) is written as:
\begin{align} \nonumber
    \ddot x + \sqrt{\frac{1}{L}} \nabla^2 f(x_1)\dot x + 2\sqrt{\mu}\dot x + (1+\sqrt{ \frac{\mu}{L}})\nabla f(x) = 0
\end{align}
For the convex case, one has:
\begin{align}\label{eq:d4} 
    \ddot x + \sqrt{\frac{1}{L}} \nabla^2 f(x_1)\dot x + 2\sqrt{\mu(t)}\dot x + (1+ \frac {\sqrt{\mu(t)}}{L})\nabla f(x) = 0
\end{align}
where $\mu(t) \rightarrow 0$ as $t \rightarrow \infty$. Since $\mu(t)$ denotes a quadratic term which tends to zero as $t \rightarrow \infty$, it can be written as: $\mu (t)=\frac{c^2}{t^2}$, thus (\ref{eq:d4}) becomes:
\begin{align}\label{eq:d5}
    \ddot x + \sqrt{\frac{1}{L}} \nabla^2 f(x_1)\dot x + 2\frac{c}{t} \dot x + (1+\frac{c}{t} \sqrt{\frac{1}{L}})\nabla f(x) = 0
\end{align}
Thus one has the similar NAIM structure as before, but value of $c$ has to be calculated. The coupling coefficient $\beta_k$ varies with each step to account for the time-varying nature of the time scale and $\beta_k \rightarrow 1$ monotonically, i.e.
$$ \beta_k=lim _{\mu \rightarrow 0} \frac{\sqrt L - \sqrt \mu}{\sqrt L + \sqrt \mu}=1$$
The derivation of $\beta_k$ is generally based on Lyapunov based argument which guraantees that the convergence rate for the convex case is $\mathcal{O} (1/k^2)$, boosted from $\mathcal{O} (1/k)$ achieved for the gradient case, i.e.
$f(x)-f(x*)\leq \frac{2L}{(k+1^2)}||x_0-x^*||_{2}^{2}$.
The two step formula for the accelerated algorithm is as before but with added difference that $\frac{1}{L}$ gets replaced by $\beta_k$, where $\beta_k \rightarrow 1$ monotonically.
\subsection{Algorithm for the convex case}
Following the notation discussed in Section \ref{sec:4}:
\begin{equation}\label{eq:d5}
    x_{k+1}=(1-\theta_k)x_k+\theta_k v_k, \theta_k<1, \forall k \in N
\end{equation}
Since $\theta_k <1$, it is convergent. Additionally $\theta_kv_k$ is the driving form and since the perturbation is decaying, $v_k \rightarrow 0$.
$y_k$ is written as a convex combination of $x_k$ and $v_k$.
\begin{equation}\label{eq:d6}
    y_{k}=(1-\theta_k)x_k+\theta_k v_k
\end{equation}
The accelerated algorithm takes the form:
\begin{align}\label{eq:d7}
    x_{k+1} = y_k - t \nabla f(x_k), t=1/L
\end{align}
\begin{align}\label{eq:d8}
        y_{k+1} = x_{k} + \beta_k(x_{k} - x_{k-1}) 
\end{align}
From (\ref{eq:d5})-(\ref{eq:d8}), $\beta_k$ is written in terms of $\theta_k$, i.e. $\beta_k=\theta_k(\theta_{k-1}^{-1}-1)$, and results in:
\begin{align} \nonumber
        y_{k} = x_{k} + \theta_k(\theta_{k-1}^{-1} -1)(x_k-x_{k-1}) 
\end{align}
Note: In the strongly convex case, one has the estimating sequence $\phi_k$ and with the fixed time scale $\sqrt(\frac{\mu}{L})$ and the condition $\phi_k\geq f(x_k), \forall k \in N,$ one could arrive at the optimal value of $\frac{1}{L}$. However, here since $\mu \rightarrow 0$, the above approach has to be modified as briefly explained below.\\
Using smoothness (descent lemma) and convexity of f(x), one arrives at the inequality:
\begin{align}\nonumber
  \frac{1}{\theta_k^2}[f(x_{k+1})-f(x_k)] +\frac{1}{2}||x^*-v_{k+1}||_{2}^2 \leq \frac{(1-\theta_k)}{L\theta_k^2}(f(x_k)-f(x^*))+1/2||x^*-v_{k}||_{2}^2
  \end{align}
 
  If $\frac{(1-\theta_k)}{\theta_k^2}\leq \frac{1}{(\theta_{k-1]^2})}$, the above equation becomes:
\begin{align} \nonumber
  \frac{1}{L\theta_k^2}[f(x_{k+1}-f(x_k^*)] +\frac{1}{2}||x^*-v_{k+1}||_{2}^2\leq \\ \frac{1}{L\theta_{k-1}^2} (f(x_k)-f(x^*))+\frac{1}{2}||x^*-v_{k}||_{2}^2 \nonumber
  \end{align}  
implying that the quantity $v_k=\frac{1}{L \theta_{k-1}^2} ((f(x_k)-f(x^*))+\frac{1}{2}||x^*-v_{k}||_{2}^2)$ is non increasing with k. Thus:
$v_{k+1}\leq v_k\leq v_{k-1}...\leq v_1$, which implies that:
$$ f(x_k)-f^* \leq \frac{L}{2} \theta_{k-1}^{2} (||x^*-x_{0}||_{2}^2) $$ and with $\theta_{k-1}=2/{k+1}$, one obtains the desired rate $\mathcal{O}(\frac{1}{k^2})$, i.e.\\
$f(x_k)-f(x^*)\leq \frac{2L}{(k+1)^2}||x_{0}-x^*||_{2}^2)$.
Thus: $$\beta_k=\theta_k(\theta_{k-1}^{-1}-1)=\frac{2}{k+2}(\frac{k+1}{2}-1)=\frac{k-1}{k+2}, k=1,2..$$ or $\beta_k=\frac{k}{k+3}, k=0,1,..$\\
The Nesterov algorithm for the smooth convex case is of the form:
\begin{align}\label{eq:d11}\nonumber
    x_{k+1} = y_k - \frac{1}{L} \nabla f(x_k), t=1/L
\end{align}
\begin{align}\nonumber
y_{k+1} = x_{k+1} + \frac{k}{k+3}(x_{k} - x_{k-1}), k=0,1,2..
\end{align}
The high resolution ODE of the above algorithm leads to:
\begin{equation}\nonumber
\ddot x+\frac{3}{t}\dot x+\frac{1}{\sqrt{L}}\nabla^2f(x) \dot x+(1+\frac{3}{2 \sqrt L})\nabla f(x)=0
    \end{equation}
Comparing with (\ref{eq:d5}), one has $c=3/2$. The details of the above derivations and its variations can be found in \cite{bubeck2015convex} and the literature dealing with Nesterov's accelerated methods.
\subsection{Triple momentum method and its variations}
By introducing a suitable perturbation dynamics along the manifold $M_0$:\noindent
\begin{align}\nonumber
(\ddot x+\frac{1}{L}\nabla^2 f(x)\dot x+\mu \dot x+ \frac{\mu}{L}\nabla f(x))+(pert)=0\\
(\ddot x+\frac{1}{L} \nabla^2 f(x)\dot x+\mu \dot x + \frac{\mu}{L} \nabla f(x))+ (\gamma \nabla^2 f(x) \dot x+\nabla f(x))=0 \nonumber\\
\implies (\ddot x+(\frac{1}{L}+\gamma) \nabla^2 f(x) \dot x+\mu \dot x + (1+ \frac{\mu }{L}) \nabla f(x))=0 \nonumber
\end{align}
Viewing the above dynamics in the form:
\begin{align}
    \dot x_1=x_2 \nonumber\\
    \dot x_2=u_1+u_2 \nonumber\\
    with\hspace{.5cm} u_1=-\frac{1}{L} \nabla^2 f(x_1)\dot x_1-\mu \dot x_1 - \frac{\mu}{L} \nabla f(x_1) \nonumber\\
    u_2=-\gamma \nabla^2 f(x_1) \dot x_1-\nabla f(x_1) \nonumber
    \end{align}
The above equations represent a linear control system, of the form:
\begin{align}\label{eq:d14}\nonumber
  \begin{bmatrix}
 \dot x_1 \\ \dot x_2 \end{bmatrix} =\begin{bmatrix}   0& 1\\
                 0& 0
\end{bmatrix} \begin{bmatrix}
    x_1\\x_2
\end{bmatrix}+ \begin{bmatrix}
    0\\1 
\end{bmatrix} u_1+ \begin{bmatrix}
    0\\1
\end{bmatrix}u_2
\end{align}
The role of $u_2$ can be understood by substituting $u_1=0$, leading to:
\begin{align}\nonumber
    \dot x_1=x_2\\
    \dot x_2=-\gamma \nabla^2f(x_1)-\nabla f(x_1)  
\end{align}
Analysis shows that the equilibrium condition for the system is:
$\dot x_1=-\frac{1}{\gamma \nabla^2f(x_1)} \nabla f(x_1)$ and this corresponds to the damped Newton equation and it converges to $x_1^*$, scaling to the equilibrium $(x_1^*,0)$ as desired.\\
Note: $(x_1^*,0)$  is the equilibrium with $u_1$ defined and $u_2=0$. Thus $u_1, u_2$ speed the approach of the system dynamics to $(x_1^*,0)$. Since the perturbation is along $M_0$, it speeds up the tangential dynamics of NAIM with reference to the transverse dynamics, leading to the process of time scale mixing. The overall algorithm then takes the form:
\begin{align}\label{eq:d14}\nonumber
    x_{k+1}=x_k+\beta(x_k-x_{k-1})-\alpha \nabla f(y_k)\\
    y_{k+1}=x_{k+1}+\nu (x_{k+1}-x_k)
\end{align}
The additional term in the dynamics of $M_0$ (i.e. $x_k+\beta (x_k-x_{k-1})$) speeds up the iterates of $x_k$ along $M_0$, thus leading to the further speeding of the convergence of $f(x_k)$ to $f(x_k^*)$. The parameters $\beta, \alpha, \nu$ are found by using automated methods as explained earlier.\\
Thus it can be summarized that the dynamics of first order optimization methods is based upon two main constructs: the NAIM structure with the slow-fast dynamics and a suitably defined decaying perturbation, which leads to time scale mixing as discussed. \\
\section{Conclusions}\label{sec:5}

One can postulate that any first order accelerated dynamic flow has two components to it: 
\begin{enumerate}
    \item NAIM with spectral gap between the dynamics of the slow manifold $M_0$ and the faster dynamics transverse to it.
    \item Acceleration is achieved using a suitable decaying perturbation to the above structure so as to obtain time scale mixing, resulting in overall accelerated dynamics.
    \end{enumerate}
On the algorithmic side, this translates to:
\begin{align}\label{eq:c1} 
        x_{k+1}=g_1(x_k,y_k,\nabla f(y_k))
  \end{align}      
  \begin{align}\label{eq:c2} 
      y_{k+1}=g_2(x_{k+1},x_k)
  \end{align}
In Nesterov's case, the R.H.S. of (\ref{eq:c1}) is dependent only on $y_k, \nabla f(y_k)$, and thus one has the form:
\begin{align}
x_{k+1} = y_k - \frac{1}{L} \nabla f(y_k) \nonumber\\ 
 y_{k+1} = x_{k+1} + \beta_k(x_{k+1} - x_k) \nonumber
\end{align}
However, the $f$ dynamics along the manifold $M_0$ can be further speeded up if $g_1()$ is a function of $x_k$ too. This is achieved if the perturbation is introduced along $M_0$. Speeding up the tangential dynamics beyond the rate at which the tangential dynamics approach $M_0$ leads to temporary (decaying perturbation breakdown of the NAIM structure) resulting in time scale mixing. This is possible in strongly convex case as the curvature of $M_0$ is defined by the Hessian term appearing in the continuous time formulation of the accelerated flows and is central to the existence of the NAIM structure (and decoupling of fast and slow dynamics).\\
The process of obtaining a faster time scale for accelerated flows is obtained by the introduction of a suitable (decaying) perturbation, which leads to transient mixing of time scales.  The ensuing dynamics on the invariant manifold $M_0$ is described in terms of its inner and outer layer dynamics, resulting in the defining the two step algorithm which characterizes accelerated gradient algorithms.\\
During the transient phase of mixing of time scales, the monotonic behaviour of the dynamics is lost and oscillatory behaviour emerges. hence, the Nesterov's accelerated optimization methods are not monotonic in nature.\\
Since the Heavy-Ball dynamics does not support a NAIM structure, the invariant manifold on which the gradient flow evolves is not structurally stable and could degenerate into an invariant set on the introduction of perturbation, resulting in the emergence of cycles/ non-convergence of the flow. The Fenichel's theorem which certifies the structural stability of the perturbed manifold is not applicable here owing to the non-existence of the NAIM structure (no pathway), owing to the absence of the connection term defined by the Hessian of $f$. This is especially true for the strongly convex functions.
This work combines the dynamical systems perspective and the $P \&I $ approach to explain the continuous-time ODE representation of the NA. The importance of normal hyperbolicity and reasons why the Heavy-Ball algorithm oscillates has been highlighted. The said approach lays foundations for ideas of time scale mixing, justifying the faster convergence of Nesterov. The mystifying NA and the tedious calculations and the involved variables and functions have been given a physical significance using the proposed approach, finally arriving at the Nesterov's iterative equations from the continuous-time perspective.

\section{Appendix}\label{sec:a}
\section*{A.1: Time Scale and Step size}

To facilitate this correspondence we start by establishing the relevant connections between the gradient descent method 
\begin{align}\label{eq:a11}
    x_{k+1} = x_k -\frac{1}{L} \nabla f(x)
\end{align}
and it’s continuous time counterpart the gradient flow system
\begin{align}\label{eq:a12}
  \dot x(t) = -\nabla f(x) : x \epsilon R^n  
\end{align}
The gradient descent is obtained by sampling the trajectory of gradient flow (\ref{eq:a12}) at times scale $\frac{1}{L}$. The step size $\frac{1}{L}$  in gradient descent directly corresponds to the time scale at which samples are drawn of the trajectories of gradient flow.
A larger step size corresponds to observing the gradient flow at more widely spaced time points  while a smaller step size gives a fine grained approximation of the continuous dynamics.
This connection between step size and time scale gives an insight into many optimization phenomena.\\
If one introduces a weight or a scaling factor $\frac{1}{L}$ into the gradient flow equation (\ref{eq:a12}) i.e. 
\begin{align}\label{eq:a13} 
    \dot x(t) = -\frac{1}{L} \nabla f(x(t))
\end{align}
then this is equivalent to a time scaling that is for the variable $\tau = \frac{1}{L} t$, one has:
\begin{align}\nonumber
    \frac{dx(\tau /\frac{1}{L})}{d\tau} = -\nabla f(x(\frac{\tau}{\frac{1}{L}}))
\end{align}
Generally written as: $\frac{dx}{d\tau} = -\nabla f(x)$ i.e.,
\begin{align}\label{eq:a15} \nonumber
   \frac{dx}{dt} = -\frac{1}{L} \nabla f(x) or \frac{dx}{d\tau} = -\nabla f(x)
\end{align}
This perspective of the equivalence of weighting in gradient flow (\ref{eq:a13}) with the step size in gradient descent (\ref{eq:a11}) established through the notion of time scale helps one to understand the optimization algorithm by unifying them with the concepts from dynamical systems theory.\\
Thus, in general, \textbf{ weight($\frac{1}{L}$) $\equiv$ step size($\frac{1}{L}$)$\equiv$ times scale $(T)$}\\
Thus faster time scale T means a larger step size means a larger weight $\frac{1}{L}$. 
The weighted gradient perspective elegantly connects continuous dynamics with discrete optimization algorithms and highlights the fundamental role of step size as both a time scaling parameter and the gradient weighting factor.\\
\textbf{Step size in Gradient Descent:}\\
For an L-smooth function, the step size in gradient descent is chosen as $\frac{1}{L}=\frac{1}{L}$ to guarantee convergence and ensure that each step makes sufficient progress without an overshoot.
If $\frac{1}{L} \leq \frac{1}{L}$ each update is guaranteed to decrease the function value and the sequence converges. Hence for this reason the optimum choice of the step size is limited to $\frac{1}{L} = \frac{1}{L}$.

\section*{A.2: P\&I approach}
In this section, we address the challenges associated with control design of both linear and nonlinear systems by viewing the system as a dynamical system modeled by an ordinary differential equation and by bringing the ideas from differential geometry to facilitate the control design. The method provides a better insight into the system behavior and a wide spectrum of issues related to design and control can thus be addressed. \\
In developing a control law for a dynamical system a simplification of the analysis and design process can be achieved if the system dynamics can be decomposed into  “fast dynamics” and some “desired slow dynamics” with the fast dynamics being associated directly with the control action. This simplifies the design process, as the control design problem can instead be framed as developing a methodology that guarantees the attractivity of the slow manifold. Though appealing, the main question associated with the above process is that how such a separation of the dynamics for a given system can be achieved.\\
Borrowing ideas from the geometric structures of dynamical systems theory, the control problem is posed in the right mathematical framework. Notions of fiber bundles, importance of the EC and splitting of the tangent space with respect to the EC, the system dynamics and control on fiber bundle structures, are highlighted in the subsequent sections. Having the sufficient background, deriving the control law through the constructive approach of the $P\&I$ method is discussed.
\subsection*{A.2.1. The Geometric Framework for Dynamical Systems}
In geometric control theory, the notion of EC on a vector bundle allows for the decomposition of the vector field into horizontal and vertical spaces, laying the foundation for decoupling the control dynamics from the desired system dynamics. The detailed definitions and preliminaries are discussed below \cite{mehra2025dynamics}.
\subsubsection*{A.2.1.1. Fiber bundle}
A fiber bundle is a quadruple \textbf{(E,$\pi$, B, F)} which consists of a total space \textbf{E}, a  base space \textbf{B}, and 'fibers \textbf{F}' such that \textbf{E} locally resembles the product space \textbf{BxF}.
The product space \textbf{BxF} is called a trivial bundle over the base \textbf{B} with fibers \textbf{F}. If the fiber of the bundle are homeomorphic to a structure group (a vector space), then the bundle is called a principal bundle (vector bundle). Similarly, one can define an affine bundle, sphere bundle, jet bundle etc.
\subsubsection*{A.2.1.2. Ehresmann Connection}
In geometry, connection provides a way to transport local geometric objects like tangent vectors, along curves in a consistent manner, allowing for comparisons of local geometries at different points. In the context of fiber bundles, an Ehresmann Connection (EC) defines a parallel transport, essentially "connecting" or identifying fibers over nearby points, which allows for the concept of parallel transport on the bundle. This provides a way to define the curvature of the fiber bundle.  Specifically, in the case of fiber bundles, an EC tells us how the movement in the total space induces changes along the fibers. 

\textbf{Definition:}
Consider a bundle \textbf{(E,$\pi$, Q, F)} with a projection map $\Pi: E\rightarrow Q$ and let $T_q\Pi$ denote its tangent map at any point 'q'. We call the kernel of $T_q\Pi$ at any point 'q' the vertical space and denote it by $V_q$. An EC, $A$,  is a vector valued one-form on $Q$ that satisfies \cite{bloch2004nonholonomic} :
\begin{enumerate}
    \item $A$ is a vector valued:  $A_q:T_q Q\rightarrow V_q$ is a linear map for each point $q\in Q$
    \item $A$ is a projection $A(v_q)=v_q  \forall \hspace{0.2cm}v_q \in V_q $
\end{enumerate}
If we denote by $H_q$ or $hor_q$ the kernel of $A_q$ and call it the horizontal space, the tangent space to $Q$ is the direct sum of $V_q$ and $H_q$, i.e. we can split the tangent space $Q$ into horizontal and vertical parts, i.e., we can project a tangent vector on to its vertical part using the connection.
\paragraph{Note:} The EC creates a direct sum decomposition of the tangent bundle $TE$ as :
$$TE=H \oplus V $$
\textbf{Connection form:} An EC can be encoded as a connection 1-form $\omega$ on $E$ with values in the vertical bundle $V_E$:\\
$\omega:TE\rightarrow V_E$\\
$\omega$ is a projection\\
$ker(\omega)$=H (horizontal space in the kernel)
The concept of EC happens to be one of the most general definitions of connection in differential geometry. It applies to any smooth fiber bundle making it more versatile than other types of connections that may have other structural requirements. It differs from other types of connections in geometry primarily through its applications to fiber bundles rather than just manifolds. Unlike affine connections, which focus on parallel transport of vectors and require a manifold structure, EC define a direct sum decomposition of the tangent space into a horizontal subspace and a vertical subspace in the total space of the fiber bundle, allowing for a broader framework applicable to various structures. 
\subsubsection*{A.2.1.3. Splitting of the tangent space}
Consider an n-dimensional manifold M with tangent bundle $TM$ such that for an $p\in M,  T_pM$ has the following structure:
$$T_pM=H_p \oplus V_p i.e. H_p \cap V_p=0$$
\paragraph{Case I}
If $M$ is co-ordinatized as $(x,\lambda$) with $x\in \mathcal{R}^k, \lambda \in \mathcal{R}^{n-k}, \hspace{.1cm} k<n $, then $T_pM$ for any $p \in M$ is written as:
$T_pM=(\dot x,0) \oplus (0,\dot \lambda)$
in the Euclidean space with the metric $I_{n \times n}$, i.e.
\begin{equation} \nonumber
\left[ \begin{array}{cc}
\dot x & 0 \\
\end{array} \right]
\left[ \begin{array}{cc}
I_{k \times k} & 0\\
0& I_{n-k \times n-k}
\end{array} \right]
\left[ \begin{array}{cc}
0 \\\dot \lambda \\  
\end{array} \right]=0
\end{equation}

Hence:
$ (\dot x, \dot \lambda)=(\dot x, 0) \oplus (0, \dot \lambda)
=H \oplus V$
\paragraph{Case II}
Now, if instead of the inner product with respect to the metric $I$, we have the metric $R$, which is a symmetric degenerate rank one matrix, called the semi-Riemann metric.\\
$R_{2 \times 2}=$
$\begin{bmatrix}
     m_{11} & m_{12}\\
                 m_{21} & m_{22}
\end{bmatrix}$
, $ m_{ij} \in \mathcal{R}, m_{12}=m_{21}, m_{22}\neq0, |R|\geq0$.\\
We have the following direct sum decomposition (for ($\dot x, \dot \lambda) \in \mathcal{R}^2$)) :
\begin{equation} \nonumber
 (\dot x, \dot \lambda)=(\dot x, -m_{21}m_{22}^{-1}\dot x) \oplus (0, \dot \lambda+m_{21}m_{22}^{-1}\dot x)   
\end{equation}
 
where:
\begin{equation} \nonumber
\left[ \begin{array}{cc}
\dot x & -m_{21}m_{22}^{-1}\dot x \\
\end{array} \right]
\left[ \begin{array}{cc}
m_{11} & m_{12}\\
m_{21} & m_{22}
\end{array} \right]
\left[ \begin{array}{cc}
0 \\\dot \lambda+m_{21}m_{22}^{-1}\dot x \\  
\end{array} \right]=0
\end{equation}
We define $\nabla \Phi=m_{21}/m_{22} $ to be the connection.
\paragraph{Case III:}
A general case:
\begin{equation} \nonumber
R=\left[ \begin{array}{cc}
R_{11} \vline & R_{12}\\
\hline
R_{21} \vline & R_{22}
\end{array} \right]=
\left[ \begin{array}{cc|c}
m_{11} & m_{12} & m_{13}\\
m_{21}& m_{22}  & m_{23}\\
\hline
m_{31}& m_{32}  & m_{33}\\
\end{array} \right]
\end{equation}
where $|M|\geq0, m_{22},m_{33}\neq0$.
Then for $\dot x, \dot y, \dot \lambda \in \mathcal{R}^3$, we have the following direct sum decomposition with respect to $M$.
\begin{equation} \nonumber
 (\dot x, \dot y, \dot \lambda)=(\dot x, -m_{21}m_{22}^{-1}\dot x,- m_{31}m_{33}^{-1}\dot x -m_{32}m_{33}^{-1}\dot y) \\ \oplus
 (0,\dot y+m_{21}m_{22}^{-1}\dot x, \dot \lambda+m_{31}m_{33}^{-1}\dot x+ m_{32}m_{33}^{-1}\dot y)   
\end{equation}
It can be verified that:
\begin{equation} \nonumber
\left[ \begin{array}{cc}
\dot x \\ -m_{21}m_{22}^{-1}\dot x \\ -m_{31}m_{33}^{-1}\dot x-m_{32}m_{33}^{-1}\dot y \\
\end{array} \right]^T
\left[ \begin{array}{cc}
R_{11} & R_{12}\\
R_{21} & R_{22}
\end{array} \right]\\
\left[ \begin{array}{cc}
0 \\\dot \lambda+m_{21}m_{22}^{-1}\dot x \\  
\end{array} \right]=\\ 0
\end{equation}

\subsubsection*{A.2.1.4. Integrability of the connection:}
The EC $ (m_{21}/m_{22}=\nabla \Phi$) is a one form and it is required to be exact (integrable). This ensures that the horizontal spaces (implicit manifold) defined by $ V_H$ are Frobenius integrable, allowing transverse sections of the fiber bundles. This implies that the curvature vanishes identically implying that the connection is flat, allowing for parallel transport along curves without twisting. If a connection is not flat, then one has topological obstructions, non-existence of the implicit manifold in defining continuous feedback control law to stabilize the system. In such a case, one has to resort to discontinuous feedback or time varying control. The smooth flow on an integrable manifold $P $ can lie on an embedded manifold $ M$ of $ P$, if $ M$ is integrable, as non-integrability of $ M$ implies that the tangent spaces of $ M$ do not form a consistent distribution that supports smooth flows globally.\\
\begin{figure}[h!]
    \centering
    \includegraphics[scale=0.5]{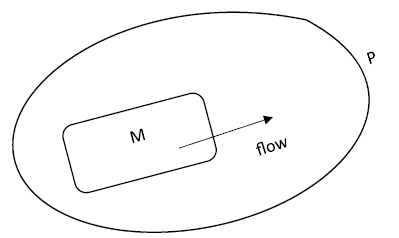}
    \vspace{0.2cm}
    \caption{An integrable submanifold $M$ within a manifold $P$}
    \label{fig:sp}
\end{figure}
\subsubsection*{A.2.1.5. Control System on Fiber bundles}
A control system can be defined in the framework of a fiber bundle as follows:\\
Let $X$ be a smooth manifold representing the state space of the control system, the total space $U$ of the bundle represents pairs of states and control, i.e., $U=X \times Y$, where $Y$ is another manifold representing the control inputs. There exists a projection map $\Pi: U\rightarrow X$, that associates each point in $U$ (a state- control pair) with the corresponding state in $X$.\\
The above structure allows for a local representation of the control system, where each fiber $U_x=\Pi^{-1} (x)$ consists of all possible control inputs that correspond to a particular state $x \in X$.\\
The dynamics of the control system which depends on the current state and the selected control input can be expressed by a smooth map $f: U \rightarrow T_X$ where $T_X$ is the tangent space of X.\\
A feedback law is defined as a continuous section $u:X \rightarrow U$ such that $\Pi \circ u=idx$. This means that for each state of $X$, there is a corresponding control input in $U$, which allows for the closed loop dynamics described by $f(u(x))$ (Figure \ref{fig:m}).
\begin{figure}[h!]
    \centering
    \includegraphics[scale=0.7]{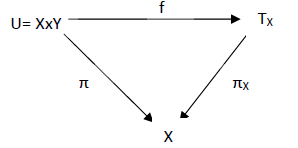}
    \vspace{0.2cm}
    \caption{Fiber bundle structure of a control system}
    \label{fig:m}
\end{figure}
\subsection*{A.2.2. The P\&I method}
Consider the control system
\begin{equation}\label{eq:oe}
\begin{gathered}
    \dot{\mathrm{x}} = f({\mathrm{x},\lambda}): ({\mathrm{x
}} ,\lambda) \in X\\
\dot{\mathrm{\lambda}} = g({\mathrm{x},\lambda})  + u 
\end{gathered}
 \end{equation} 
 where $X=(x,\lambda) \in {\mathcal{R}^{n-1}} \times \mathcal{R}$ is a n-dimensional state space.
based on the desired control objectives, we define the target dynamics: $\lambda=\Phi(x)$, i.e., $\dot x=f(x,\Phi(x))$. This implies that a control $u$ is to be determined such that all the system trajectories reach the target (implicit) manifold:  
\begin{equation} \nonumber
M_0(x,\lambda) = \{{({\mathrm{x}},\lambda) | \lambda + \Phi(x) = 0}\} , (x^*,\lambda^* \in M_0)   
\end{equation}
The steps to generate the control law are summarized as follows :
\begin{enumerate}
    \item Given the implicit manifold $M_0 (x,\lambda)$, the normal to the $M_0$ is given as: 
$$\nabla M_0 = \begin{bmatrix}
\nabla \Phi(x) & 1 
\end{bmatrix}$$
\item Define a degenerate symmetric rank one matrix $ R$ as :
\begin{align}\label{eq:a4}
  R = \nabla M_0^T \nabla M_0= \begin{bmatrix}
\nabla \Phi^T \nabla \Phi & \nabla \Phi \\
\nabla \Phi & 1 
\end{bmatrix}  
\end{align}
$R$ has the form of a semi-Riemannian metric. 
\item Using the fiber bundle structure of the control system and $R$ as defined above, we invoke the idea of EC on the fiber bundle, ($n=2$) i.e.
\begin{align}\label{eq:a5} \nonumber
  R = \nabla M_0^T \nabla M_0= \begin{bmatrix}
\nabla \Phi^T \nabla \Phi & \nabla \Phi \\
\nabla \Phi & 1 
\end{bmatrix} = \left[ \begin{array}{cc}
m_{11} & m_{12}\\
m_{21} & m_{22}
\end{array} \right]
\end{align}
 Then the connection $(\nabla \Phi=m_{21}/m_{22})$ leads to splitting of the tangent bundle $TM$ as described in next step.
 \item If $(\dot x, \dot \lambda ) \in TM$, then one has the direct sum decomposition: 
\begin{equation}\label{eq:a6} \nonumber
 (\dot x, \dot \lambda)=(\dot x, -m_{21}m_{22}^{-1}\dot x) \oplus (0, \dot \lambda+m_{21}m_{22}^{-1}\dot x)   \\
 =H_H+V_V
\end{equation}
where $H_H, V_V$ correspond to the horizontal and vertical space. At each $p\in X$, one has:\\
\begin{align}\nonumber
    T_{p}M = H_{p} \oplus V_{p}, \mid H_{p} \cap V_{p}={\emptyset}
\end{align}
Since $\dot \lambda=u$, it can be inferred from the above splitting that the control $u$ acts along the vertical space $V_V$. 

\begin{figure}[!h]
    \centering
    \begin{subfigure}{0.2\textwidth}
        \centering
        \includegraphics[width=1.1\textwidth]{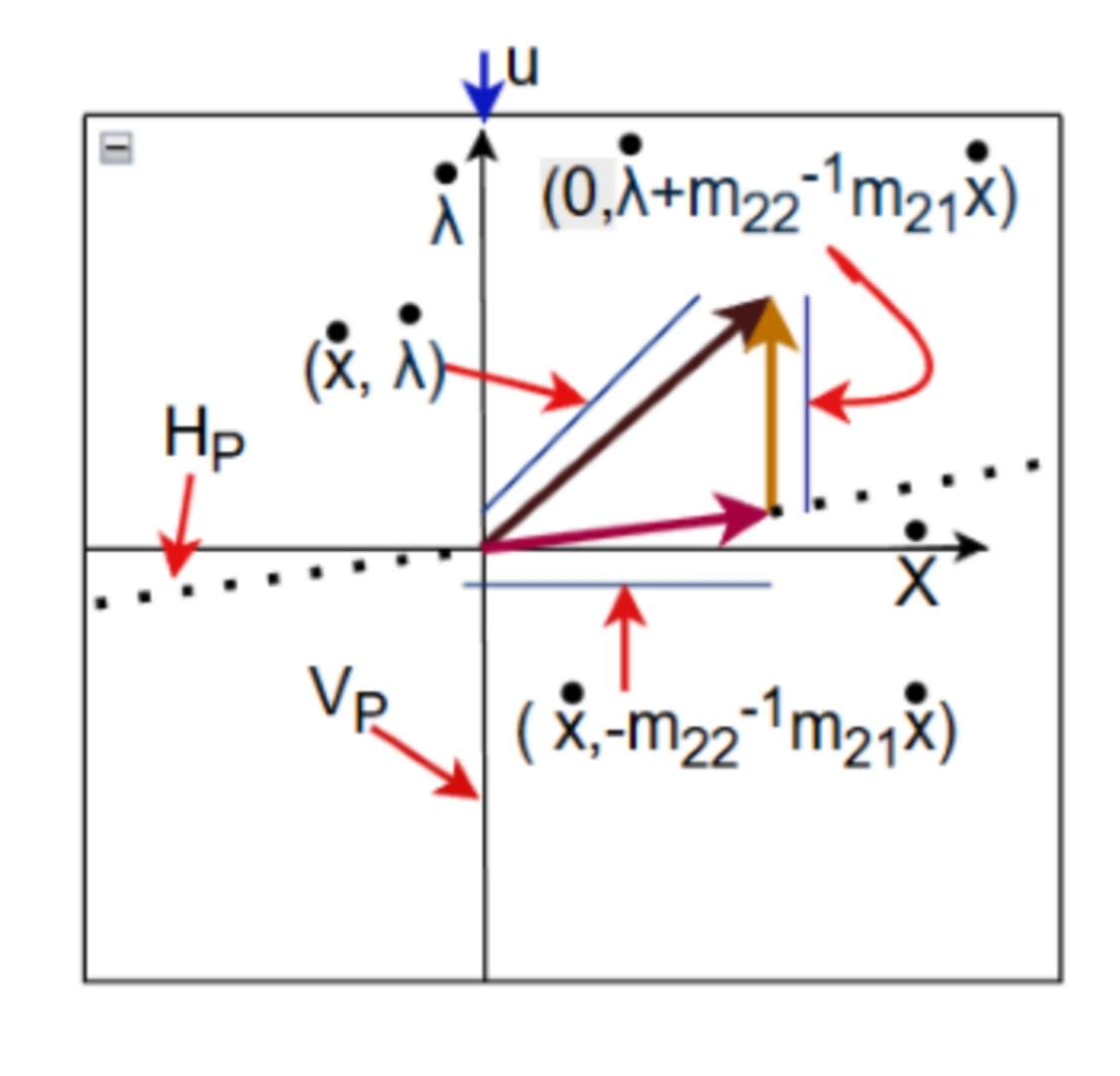}
    \end{subfigure}%
   ~ 
    \begin{subfigure}{0.2\textwidth}
        \centering
        \includegraphics[width=1.105\textwidth]{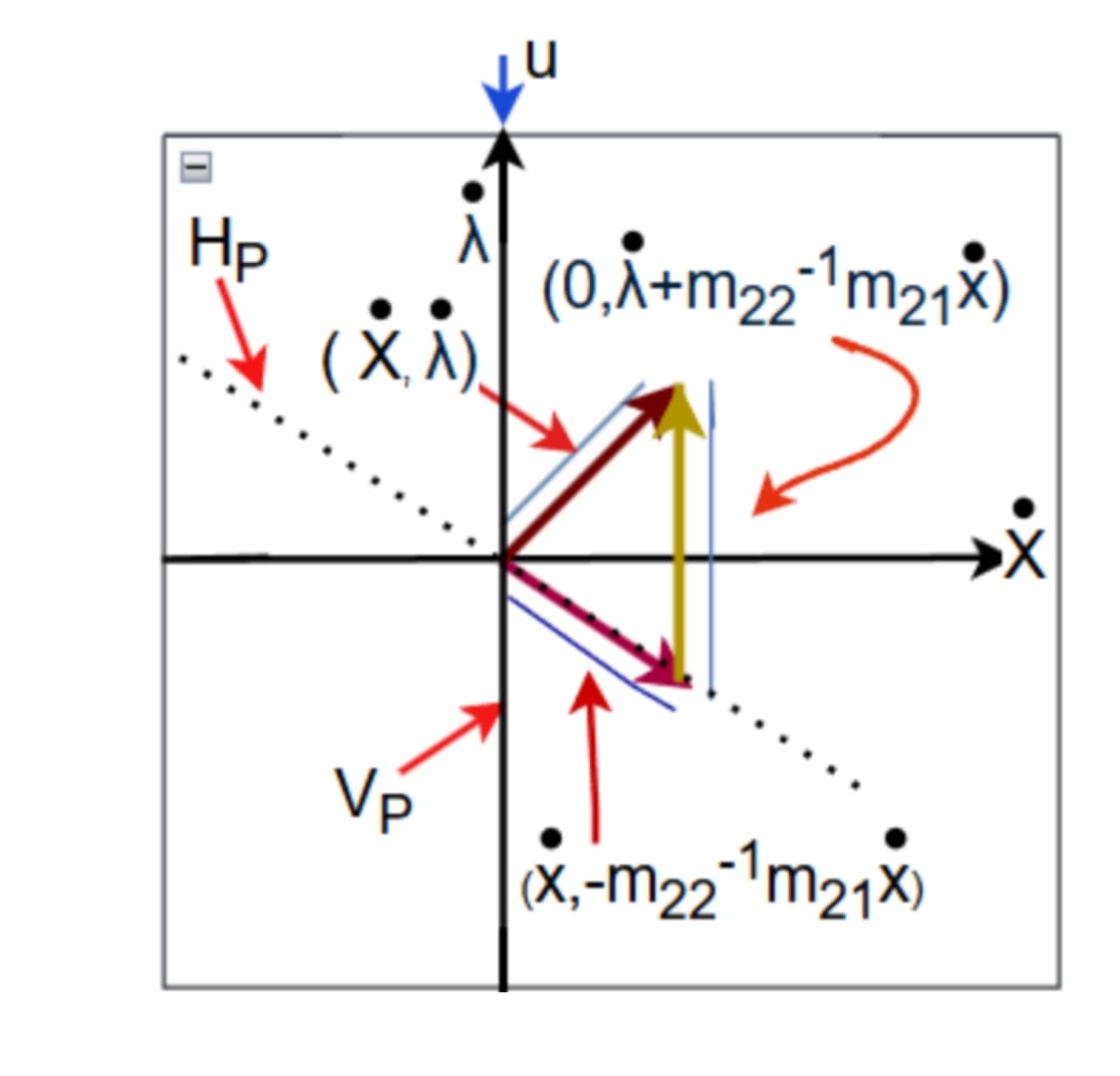}
    \end{subfigure}
    \caption{Geometrical interpretation: vertical vector $V_p$ is along the fiber direction and $H_p \oplus V_p = T_pM$}
    \label{fig:sp}
\end{figure}
\item As $u$ acts along $V_V$, the passive output $y$ is chosen as $y=(\dot \lambda+m_{21}m_{22}^{-1}\dot x)$.
If the passive output is integrable, i.e. $m_{21}m_{22}^{-1}\dot x=\nabla \Phi(x) \dot x$ is integrable,  implying $\Phi(x)$ is an integrable connection \cite{nayyer2022towards}, then one can define a storage function S:
\begin{align}\nonumber
  S=\frac{1}{2}(\int y \hspace{0.2cm}dt)^2=\frac{1}{2}(M^2)  
\end{align}
If $\Phi(x) $ is an integrable connection, where $M=\{{({\mathrm{x}},\lambda) | \lambda + \Phi(x) \neq 0}\}$, i.e. M is defined for $(x,\lambda)$ not lying on $M_0$.

\item The control law $u$ follows from the condition:
$\dot S \leq -\hat \alpha S$, it follows that: \\
$\dot M\leq -\alpha M (\alpha=\hat \alpha/2)$\\
Choosing $u$ to satisfy $\dot M\leq -\alpha M$ gives:
\begin{align}\label{eq:a8}
    u=-\nabla \Phi \dot x-\alpha M
\end{align}
So, the controlled system (\ref{eq:oe}) is of the form:
\begin{align}\nonumber
    \dot{\mathrm{x}} = f({\mathrm{x},\lambda}) 
  \end{align}  
 \begin{align}\label{eq:a10} 
\dot{\mathrm{\lambda}} = -\nabla \Phi \dot x-\alpha M
 \end{align}
Since the control law follows from the fact that $\dot M\leq -\alpha M$, it follows that $M\rightarrow TM_{0(x^*,\lambda^*)}$, where $(x^*,\lambda^*)$ are the equilibrium points belonging to $M_0$.\\

\item An alternative way to look at (\ref{eq:a10}) is as follows:
$$\dot \lambda = -\nabla \Phi\dot x - \alpha M \equiv  \dot \lambda + \nabla \Phi\dot x = -\alpha M$$
i.e. the vector field along which the control law acts is along $V_V$. Thus the controlled trajectories are along $V_V$ and orthogonal to $V_H$. We call this direction along $V_V$ the pathway defined by the connection $\nabla \Phi(x) = \frac{m_{21}}{m_{22}}$ for the off-manifold trajectories of $M_0$ to approach $M_0$ under the control action $u$.\\
\item Finally the control law (\ref{eq:a8}) can be broken down as follows: 
$$u = \underbrace{-\nabla \Phi(x)\dot x}_{pathway} - \alpha M$$
$: \nabla \Phi(x) \dot x$ - defines the pathway for the controlled trajectories which is defined by the connection $\nabla \Phi(x)$ and is along the $V_V$ direction.\\
$: \alpha M$ - this term pushes the off manifold trajectories to $M_0$.\\
\item Summarizing:
\begin{itemize}
    \item $M_0 = \{(x,\lambda) \epsilon X |\lambda + \Phi(x) = 0\}$ - Target manifold based on the desired target dynamics\\
    \item If the connection $\nabla \Phi(x)$ is integrable as discussed, then, $M = \{(x,\lambda) \epsilon X |\lambda + \Phi(x) \neq 0\}$ and the control law is defined as: \\
    $u=-\nabla \Phi(x)\dot x - \alpha M$ and $M\rightarrow TM_0|(x^*,\lambda^*)$
    \item Given the dynamical system (which is stable for $u=0$):
    \begin{align}\nonumber
        \dot x = f(x,\lambda)\\
        \dot \lambda = g(x,\lambda) + u \nonumber
    \end{align}
   The introduction of the pathway term $-\nabla \Phi \dot x$ above gives:
 \begin{align}\nonumber
        \dot x = f(x,\lambda)\\
        \dot \lambda = g(x,\lambda) -\nabla \Phi \dot x \nonumber
    \end{align}
   for some defined $M_0$ and resulting connection ensures that the system trajectories now converge to $(x^*,\lambda ^*)$ but along the pathway defined by the connection \textbf{$\nabla \Phi$}\\
    \item Finally, if $M_0$ is defined as $M_0 = \{(x,\lambda)|\lambda \pm \Phi(x) = 0\}$ then $u = \mp \nabla \Phi(x)\dot x - \alpha M$
\end{itemize}
\end{enumerate}
\textbf{Remarks:}
\begin{enumerate}
    \item The target manifold which contains the invariant set is attractive under the control action of $ u$, i.e. the off-manifold trajectory converges to the manifold exponentially.
    \item If the target manifold is itself is not invariant, i.e. the target dynamics are not tangential to it, then one can think of the invariant set contained in $M$ as a controlled normally hyperbolic manifold which is perturbed by the target dynamics which converges to it exponentially. The stability of the perturbed invariant manifold follows from the application of the Fenichel's theorem, which provides a rigorous framework for analyzing how invariant manifolds persist and behave under perturbations.
    \item The target manifold helps to set up the geometry from which the idea of NHIM and associated ideas follow.
    \end{enumerate}

\textbf{Definition:} Let $M$ be a smooth compact sub-manifold (implicit manifold M) embedded in $ R^n$, with a dynamical system described by the flow $f: X\rightarrow X$. The manifold M is normally hyperbolic at a point $ p\in M$ if the tangent space of the ambient manifold at $ p\in M$ denoted by $ T_pM$ admits a splitting that is invariant  differential of the dynamics.\\
Specifically:
\begin{enumerate}
    \item Tangent Bundle splitting \cite{kuehn2015multiple}: $$ T_pX=T_pM \oplus E_p^s\oplus E_p^u$$
where $E_p^s:$ denotes a stable bundle, corresponding to directions that contract under the dynamics (stable directions) \\
$E_p^u:$ denotes the unstable bundle corresponding to directions that expand under the dynamics.
    \item Invariance: The splitting is invariant under the differential :
    $ Df_p:T_pX\rightarrow T_{f(p)} X$
    \item Hyperbolicity condition: The dynamics in the normal directions($E_p^s,E_p^u $) dominate the dynamics in the tangential directions.
    \item Global Definition: The manifold M is normally hyperbolic if the above splitting and (other conditions) hyperbolicity holds for any $p \in M$. The tangent bundle $TM,E_p^s,E_p^u $ together form a Whitney sum over M.\\
    $ T_XM=TM \oplus E^S \oplus E^u$.
\\
\textbf{Note}
\begin{itemize} 
\item Using the EC and the implicit manifold M, one has:
$ TX= TM \oplus E^s=V_H\oplus V_V$
where the dynamics $ V_V$ are contracting under the control action of $u$ as explained in the P\&I method, i.e. under the action of $ u$, one has a controlled normally hyperbolic manifold.
\item Considering the target dynamics, the above statement holds true if the target dynamics lies tangent to the manifold $ M $ and then one has Normally Hyperbolic Invariant Manifold $M$. However, when the above does not hold and the target dynamics is asymptotically stable, then the manifold $M$ qualifies as a normally hyperbolic manifold containing the invariant set that is the equilibrium point which is viewed as being perturbed by the target dynamics.
\end{itemize} 
\item Fenichel's theorem \cite{wiggins2013normally}: (Fast-slow version): The essence of the theorem is as follows: \\
Suppose $S$ is a compact normally hyperbolic sub manifold (equilibrium point) of the implicit manifold M, then  for $\epsilon>0$ sufficiently small, there exists a locally invariant manifold $S_\epsilon$ diffeomorphic to S. Local invariance implies that the trajectories can enter or leave $ S_\epsilon$ only through its boundaries. \\
Note: If the implicit manifold is invariant, then by the application of the above theorem, it implies that $M$ is perturbed to a locally invariant manifold $M_\epsilon$ diffeomorphic to M.

\end{enumerate}

\printbibliography
\end{document}